\documentclass[pdftex,bookmarks=true,colorlinks=false,
final
]{dmtcs-episciences}

\begingroup
  \def\x{\endgroup\ExecuteOptions{dvipdfm}}%
  \ifx\pdfoutput\undefined
  \else
    \ifx\pdfoutput\relax
    \else
      \ifcase\pdfoutput
        
      \else
        \def\x{\endgroup\ExecuteOptions{pdftex}}%
	
      \fi
    \fi
  \fi
\x%

\usepackage[utf8]{inputenc}
\usepackage{subfig}
\usepackage{enumerate}
\usepackage{graphicx}
\usepackage{amsfonts}
\usepackage{amsmath}
\usepackage{amssymb}
\usepackage{marginnote}
\usepackage[normalem]{ulem}
\usepackage{marginnote}
\usepackage{mathtools}
\usepackage{longtable}
\usepackage{ulem}
\usepackage{cancel}
\usepackage{url}
\usepackage{hyperref}    
\usepackage{array}
\usepackage[boxed,vlined,linesnumbered,noend]{algorithm2e}
\usepackage{subfig} 
\usepackage{float}
\usepackage{pifont}
\usepackage{soul}

\newtheorem{definition}{Definition}
\newtheorem{fact}{Fact}
\newtheorem{claim}{Claim}
\newtheorem{theorem}{Theorem}
\newtheorem{observation}{Observation}
\newtheorem{lemma}{Lemma}

\SetKwInOut{Input}{Input}
\SetKwInOut{RestInput}{}
\SetKwInOut{Output}{Output}
\SetKwInOut{RestOutput}{}
\SetKwInOut{Actors}{Actors}
\SetAlFnt{\small}
\input{commands}

\usepackage[round]{natbib}

\author[Dominik Bojko et al.]{Dominik Bojko\affiliationmark{1}
  \and Krzysztof Grining
  \and Marek Klonowski\affiliationmark{2}}
\title[Probabilistic Counters for Privacy Preserving Data Aggregation]{Probabilistic Counters for Privacy Preserving Data Aggregation\thanks{Supported by Polish National Science Center grant number UMO-2018/29/B/ST6/02969UMO-2018/29/B/ST6/02969}}

\affiliation{
  Wrocław University of Science and Technology, Faculty of Fundamental Problems of Technology, Department of Fundamentals of Computer Science, Wrocław, Poland\\
  Wrołcaw University of Science and Technology, Faculty of Fundamental Problems of Technology, Department of Artificial Intelligence, Wrocław, Poland}
\begin{document}



\maketitle

\publicationdata{vol. 28:2}{2026}{14}{10.46298/dmtcs.11614}{2023-07-20; 2023-07-20; 2026-02-12}{2026-02-12}

\begin{abstract}

Probabilistic counters are well-known tools often used for space-efficient set cardinality estimation. In this paper, we investigate probabilistic counters from the perspective of preserving privacy. We use the standard, rigid differential privacy notion. The intuition is that the probabilistic counters do not reveal too much information about individuals but provide only general information about the population. Therefore, they can be used safely without violating the privacy of individuals. However, it turned out, that providing a precise, formal analysis of the privacy parameters of probabilistic counters is surprisingly difficult and needs advanced techniques and a very careful approach.

We demonstrate that probabilistic counters can be used as a privacy protection mechanism without extra randomisation. That is, the inherent randomisation of the protocol is sufficient to protect privacy, even if the probabilistic counter is used multiple times. In particular, we present a specific privacy-preserving data aggregation protocol based on Morris Counter and MaxGeo Counter. Some of the results presented are devoted to counters that have not been investigated so far from the perspective of privacy protection. Another part is an improvement of the previous results. We show how our results can be used to perform distributed surveys and compare the properties of counter-based solutions and a standard Laplace method.

\end{abstract}

\section{Introduction}\label{sect:intro}

Since Big Data related topics have been widely developed in recent years, solutions that focus on saving memory resources have become very popular. We would like to consider a standard example of such space-efficient mechanisms, namely probabilistic counters, which are used to represent the cardinality of dynamically counted events. More precisely, we would like to indicate the occurrence of $n$ events using a very small (significantly less than $\log n$) number of bits. We assume that $n$ is unknown in advance and may change. Clearly, a simple information-theoretic argument convinces us that it is not feasible if we demand an exact representation of the number of events. Nevertheless, there are some very efficient solutions that require only $\Theta(\log\log n)$ bits and guarantee sufficient accuracy for a wide range of applications. As examples, one can point most \textit{probabilistic counters} --  probabilistic structures well known in the literature since the seminal Morris' paper~\cite{morris1978counting} followed by its thorough mathematical analysis by Flajolet in~\cite{flajolet1985approximate}. They are used as building blocks in many space-efficient algorithms in the field of data mining or distributed data aggregation in networks   or smart metering, just to mention a few (\cite{Baquero2009} or \cite{JCIGotfryd}).

In this paper, we investigate probabilistic counters from the privacy-protection perspective. Our analysis is based on a differential privacy notion, which is commonly considered the only state-of-the-art approach. The differential privacy has the undeniable advantage of being mathematically rigorous and formally provable, contrary to previous anonymity-derived privacy definitions. This approach to privacy-preserving protocols can be used to give a formal guarantee for privacy that is resilient to any form of post-processing. For a survey about differential privacy properties, see \cite{DworkAlgo} and references therein. Analysis of protocols based on differential privacy is usually technically complex, but by using this notion, we are immune to, e.g., linkage attacks (see, for example \cite{narayanan2009anonymizing,narayanan2010myths}). 

The idea behind differential privacy is as follows: for two "neighbouring" scenarios that differ only in the participation of a single user, a differentially private mechanism should provide a response chosen from very similar distributions. Roughly speaking, differential privacy is described by two parameters: $\varepsilon$ -- which controls a similarity of probabilities of common events -- and $\delta$ -- which is related to a probability of unusual events. The smaller the parameters, the better from the privacy point of view. In effect, judging by the output of the mechanism, one cannot say if a given individual (user) was taken into account for producing a given output. Intuitively, probabilistic counters should provide a high level of differential privacy since, statistically, many various numbers of events are "squeezed" into a small space of possible output results. In the case of one counter considered in our paper (MaxGeo) counter, one can find some similar, recent results about the privacy the algorithm offers. Nevertheless, the question about the value of the parameters of the potential differential privacy property remains open (see the discussion in Section~\ref{sect:related}).
Moreover, when considering a small number of events $n$, an additional problem may occur, as it may be possible to distinguish that the number of events is different from $n-1$ or $n+1$ with a significant probability.

 In our paper, we provide a very precise analysis of two well known probabilistic counters from the perspective of preserving privacy. It turned out that this task is surprisingly complex from the mathematical point of view.
Our primary motivation is to find possibly accurate privacy parameters for the two most fundamental probabilistic counter protocols, namely the Morris Counter~\cite{morris1978counting} and the MaxGeo Counter~\cite{szpankowski1990yet}. Note that the second one is used for yet another popular algorithm --- HyperLogLog~\cite{flajolet2007hyperloglog}. One may realise that these two counters are relatively old; however, they, together with their modifications, have been extensively used until these days. Morris Counter is often used in big data solutions, for instance, to measure network's capabilities \cite{ICEBuckets}. The most crucial examples of refinements of the HyperLogLog algorithm are mentioned in Section~\ref{sect:related}.

 We claim that a high-precision analysis in the case of probabilistic counters is particularly important. This is because even a mechanism with very good privacy parameters can cause a serious privacy breach when used multiple times. That is, privacy loss/information leakage accumulates over multiple releases (see, e.g., \cite{DworkAlgo}).  
 Probabilistic counters in realistic scenarios may be used as fundamental primitives and subroutines in more complex protocols, since the differential privacy property is immune to post-processing.

We also show that those two probabilistic counters can be used safely without any additional randomisation, even in very demanding settings. It is commonly known that no deterministic algorithm can provide non-trivial differential privacy. However, Probabilistic counters have inherent randomness, achieving the desired privacy parameters. In other words, one can say that probabilistic counters are safe by design, and we do not need any additional privacy-orientated methods. In particular, what is most important, existing, working implementations do not need to be changed if we start demanding the provable privacy of a system.
 
Finally, we demonstrate how our results can be used for constructing a data aggregation protocol based on probabilistic counters that can be used in some specific scenarios until we want them to satisfy even more rigorous privacy properties.

To the best of our knowledge, most of the results are new and deal with protocols not considered before in the context of privacy preservation. Some other (such as the MaxGeo counter) improve some previous results (e.g.~\cite{ADAMS}).

\subsection{Paper structure and results}
Starting from this point, for the sake of clarity, we use the abbreviation 
\textit{DP} as a shortcut for differential privacy, while this property is described by some parameters. 
 
The main contribution of our paper is as follows:

\begin{itemize}
\item We prove that the classic Morris Counter satisfies $\left(\varepsilon(n),\delta(n)\right)$-DP with \\
$\varepsilon(n)=O\left(\frac{\left(\log(n)\right)^2}{n}\right)$ and
$\delta(n)=O\left(\max\left\{n^{-\left(\ln(n)\right)^{c-1}},n^{-1}\left(\ln(n)\right)^{-c}\right\}\right)$, for any $c>0$ (Theorem~\ref{thm:addition} in Section~\ref{results}).
\item We prove that the Morris Counter satisfies the $\left(L(n),0.00033\right)$-DP property (see Definition \ref{dpDef}), where $L(n)=-\ln\left(1-{16}/{n}\right)\approx {16}/{n}$ (Theorem~\ref{thm:main} in Section~\ref{results}). In Observation~\ref{lowerObs}, we also show that the constant $16$ cannot be improved.
\item We prove that MaxGeo Counter satisfies the property $\left(\varepsilon, \delta\right)$ -DP (Definition \ref{dpDef} is provided in Section \ref{sect:model}) if the number of events $n$ (in Section \ref{sect:counters} the concept of event is clarified) is at least $\dfrac{\ln(\delta)}{\ln\left(1-2^{-l_{\varepsilon}}\right)}$, where $l_{\varepsilon} = \left\lceil \log\left(\frac{e^{\varepsilon}}{e^{\varepsilon}-1}\right) \right\rceil$ (Theorem~\ref{maxGeoTheorem} in Section~\ref{sect:counters}).
\item We construct a distributed survey protocol to preserve privacy based on probabilistic counters in Section~\ref{sect:scenario} and compare it with the Laplace method, which is considered as the actual state of the art of differentially private protocols and is not based on probabilistic counters.
\end{itemize}

The remainder of this paper is organised as follows. First,  in Section~\ref{sect:related} we mention work related to our paper and some popular examples of other probabilistic counters, which are not considered in this paper. In Section~\ref{sect:model}, we recall the differential definition of privacy. In Section~\ref{sect:counters} we 
describe probabilistic counters; 
further, we recall the definitions of both Morris and MaxGeo Counters. Moreover, we state Fact~\ref{help}, a useful reformulation of the standard definition of differential privacy for probabilistic counters. 
In Section~\ref{sect:scenario}, we demonstrate how a probabilistic counter can be used to construct a data aggregation protocol in a very particular, yet natural scenario.
Section~\ref{results} consists of formulations of our main results for both counters.
Section~\ref{practice} gives some ideas about the realisation of the scenario and the comparison of both counters.
We also compare these solutions with the standard Laplace method (Section~\ref{practice}).
For convenience of the reader, we provide the proofs in Section~\ref{sec:proofs}, as they are rather technical. For more clarity, some proofs and lemmas are moved to \ref{append}.
Finally, in Section~\ref{sect:conclusion} we present conclusions and future work propositions.

 \section{Previous and Related Work}\label{sect:related}

In our paper, we provide a detailed analysis of some 
probabilistic counters from the perspective of differential privacy. Differential privacy concepts have been discussed in many papers in recent years. One can also find a well developed body of literature devoted to probabilistic counters and similar structures. Therefore, we limit the related literature review to the most relevant papers.    

\paragraph{Differential Privacy literature}
In our paper, we focus on the inherent privacy guarantees of some 
probabilistic structures defined as differential privacy. The idea of differential privacy has been introduced for the first time in~\cite{dwork2006calibrating}; however, its precise formulation in the widely used form appeared for the first time in~\cite{Dwork06}. There is a long list of papers concerning differential privacy, e.g.~\cite{dwork2006our,dwork2009differential,dwork2010differential}, to mention a few. 

Most of these papers focus on a centralised (global) model, namely, a database with a trusted party holding it. See that in our paper, despite the distributed setting, we have the same (non-local) trust model. 
In particular, we assume an existence of a \textit{curator} that is entitled to gather and see all participants' data in the clear and release the computed data to a wider (possibly untrusted) audience. Comprehensive information on differential privacy can be found in~\cite{DworkAlgo}.

\paragraph{Probabilistic counters and their applications}


The idea of probabilistic counters, along with the well-known Morris Counter was presented in the seminal paper~\cite{morris1978counting}. The aim was to construct a very small data structure to represent a large set of events of some kind. In our paper, we focus on the Morris Counter analysed in detail in~\cite{flajolet1985approximate}.
The second structure discussed in our paper is MaxGeo Counter, introduced and analysed in~\cite{szpankowski1990yet}. More detailed and precise analysis can be found in~\cite{eisenberg2008expectation}. The most important application of MaxGeo Counter can be found in~\cite{flajolet2007hyperloglog}, where the authors propose the well-known HyperLogLog algorithm. Its practical applications are widely described in~\cite{heule2013hyperloglog}. There are several widely used improvements of the HyperLogLog algorithm: HyperLogLog+~\cite{heule2013hyperloglog}, Streaming HyperLogLog with sketches based on historical inverse probability~\cite{StreamHIP} or martingal estimator \cite{StreamHLL} or empirically adjusted HyperBitBit (proposed by R.~Sedgewick \cite{HyperBitBit}). The main goal of these adjustments is to reduce the memory requirements (see, e.g., \cite{FastRAQ} or \cite{TingBillionsDatasets}).
For instance, some of the above solutions are used in database systems for query' optimisation or for document classification purposes.
Moreover, the MaxGeo counter was used in \cite{ANF2002}, for an adjustment of the ANF tool, developed for data mining from extensive graphs, which enables it to answer many different questions based on some neighbourhood function defined on the graph.


Unsurprisingly, one of the main applications of the approximate counter is to compute the size of a database or its specific subset. A set of such applications can be found in~\cite{flajolet1985probabilistic}. In~\cite{van2009probabilistic}, the authors use Morris Counter for online, probabilistic, and space-efficient counting over streams of fixed, finite length. The authors of \cite{cichon2011approximate} proposed an application of a Morris Counter system for flash memory devices. Another application, presented in \cite{2009arXiv0904.3062C}, is a revisit of Morris Counter designed for binary floating-point numbers. In \cite{Gronemeier2009}, Morris Counter is used in a well-known problem of counting frequency moments of long data streams. The authors of~\cite{dice2013scalable} focused on making probabilistic counters scalable and accurate in concurrent settings. The paper on probabilistic counters in hardware can be found in~\cite{riley2006probabilistic}.
A slightly modified version of Morris Counter called Morris+ was recently introduced in~\cite{nelson2022optimal} with the proof of its optimality in terms of accuracy--memory trade-off.

In random graphs theory, Morris Counter is usually connected to greedy structures. For instance, in an arrangement of a randomly labelled graph in Gilbert model $G(n,p)$, it is possible to construct a greedy stable set $S_n$, which size has the same distribution as the Morris Counter $M_n$ of the base $a=(1-p)^{-1}$ (see, e.g., \cite{frieze_karonski_2015} or \cite{bollobas1998random} for the fundamentals of random graph theory).

There are many other birth processes that are quite similar to the Morris Counter, which are applicable in a variety of disciplines like biology, physics, or the theory of random graphs. Short descriptions of such examples can be found in \cite{Crippa:1997:QMP:2781893.2781980}. When talking about probabilistic counters, it is worth mentioning the Bloom filter~\cite{Bloom1970}, which are space-efficient probabilistic data structures that are representations of sets. There exists a probabilistic counter that approximates the number of elements represented by the given Bloom filter \cite{DBLP:journals/jcisd/SwamidassB07a}.

Other common examples of probabilistic counters are $F_p$ counters \cite{ALON1999137,DBLP:conf/focs/Indyk00}, which approximate the $p$-th moments of frequencies of occurrences of different elements in the database. Let us also mention a paper \cite{Kamil2} in which one can find numerous applications of similar constructions to create pseudorandom sketches in Big Data algorithms.

Notice that the variety of possible applications of probabilistic counters creates an opportunity to exploit inherent differential privacy properties. However, a new challenge arises --- to calculate the parameters of differential privacy for those counters, which are not connected straightforwardly with Morris or MaxGeo Counters.

\paragraph{Probabilistic counters and preserving privacy}

Some probabilistic counters and similar structures were previously considered in terms of privacy preservation. We mention only the papers strictly related to the algorithms discussed in our paper (i.e., Morris Counter and MaxGeo).
The authors of~\cite{desfontaines2019cardinality} show that in the scenario of using different types of probabilistic counters for set cardinality estimation with the Adversary being able to extract the intermediate values of the counter, privacy is not preserved. Note that in this paper, we perform data aggregation instead of cardinality estimation. Moreover, we assume the Adversary is not able to extract any intermediate values from the counter. That is, we consider a \textbf{global} model, while the result from \cite{desfontaines2019cardinality} assumes the settings closer to the classic local model~\cite{DworkAlgo}.

One of the main results of this submission is a careful and tight analysis of Morris Counter from the context of preserving privacy. To the best of our knowledge, such analises have not been provided so far. Our second contribution is an analogous analysis of the MaxGeo. There are a few very recent 
papers presenting privacy-preserving protocols that use the Flajolet--Martin sketch as a building block.
 In~\cite{flajolet1985probabilistic}, so called Flajolet---Martin sketch was introduced. In~\cite{ADAMS} authors consider a general concept of a probabilistic counter, based on several MaxGeo counters and its differential privacy. Nevertheless, they incorrectly call it a generalisation of the Flajolet--Martin sketch (they probably confuse the Flajolet--Martin sketch with LogLog sketch).
 
We concern a concept of MaxGeo counter, which is a core of LogLog or HyperLogLog sketch, however, it can be used in other arrangements as well. 
These papers in some cases provide an analysis of the privacy guaranteed by Flajolet---Martin with the global model. In all the cases, the conclusion is positive in the sense that the protocol itself provides some level of differential privacy without adding extra randomness. 
Beneath papers provide an analysis of privacy guaranteed by the sketches related to LogLog algorithm, which can be seen as the processing of the fundamental MaxGeo counters. From this point of view, our contribution about MaxGeo counters has a larger applicative potential.

 In \cite{ADAMS} the authors consider, among others, a sketch that can be seen as a particular 
application of the MaxGeo counter. They introduce its differentially private version via 
trick (adding artificial utilities) and provide its accuracy when used to count the number of elements in multisets. Accidentally, a proof of the basic theorem from \cite{ADAMS} uses an incorrect argument (inappropriate utilisation of Hoeffding's inequality), so it is difficult to compare the results precisely. Nevertheless, the overlap of results between our paper and \cite{ADAMS} is only partial.

In \cite{PanChoi}, the authors consider the LogLog sketch as a subroutine. After a careful analysis, they show that it is asymptotically $(\varepsilon,\delta=\mathrm{negl}(\lambda))$-DP (with respect to the numbers of different elements), when the number of elements counted by the mechanism is at least $8K\lambda \max(\frac{1}{\varepsilon},1)$, where $K$ is some accuracy parameter, $\lambda$ is some security parameter and $\mathrm{negl}(x)$ is some negligible function of argument $x$ (Theorem 4.2 in~\cite{PanChoi}). Nevertheless, the analysis does not explain how to choose parameters $K$ and $\lambda$ in order to obtain $(\varepsilon,\delta)$-DP for a given $\varepsilon$ and $\delta$ parameters. Moreover, a consideration of asymptotic behaviour (with respect to the number of unique elements $n$) is not relevant when the hash function restricts the possible result to the size bounded by its domain. Our analysis of the MaxGeo counter provides an exact (non-asymptotic) dependence between $n$ and the parameters $\varepsilon$ and $\delta$.

We also mention that some other pseudorandom structures have been analysed from the perspective of differential privacy. For example, in~\cite{AniaPio}, the authors considered Bloom filters as a means of constructing a privacy-preserving aggregation protocol.


\section{Differential Privacy Preliminaries}\label{sect:model}

In this section, we briefly recall \textit{differential privacy}. For more details, see, e.g.~\cite{DworkAlgo}. We denote the set of (positive) natural numbers by $\NN$ and the set of all integers by $\ZZ$. Moreover, let $\NN_0=\NN\cup\{0\}$. For $a,b\in\ZZ$ let us define a discrete interval $[a,b]\cap\ZZ$ by $[a:b]$.
We also define $[n]=\{1,2,\ldots,n\}$ for $n\in \NN$.
We assume that there exists a trusted \textit{curator} who holds, or securely obtains, the data of \textit{individuals} in a (possibly distributed) database $x$. 
Every row of $x$ consists of the data of some individual. By $\mathcal{X}$, we denote the space of all possible rows. The goal is to protect the data of every single individual, even if all users except one collude with an \textit{adversary} to breach the privacy of this single, uncorrupted user. On the other hand, the curator is responsible for producing a \textit{release} -- a possibly accurate response to a requested \textit{query}. This response is then released to the public, who is allowed to perform a statistical analysis on it. The differential privacy is, by design, resilient to post-processing attacks, so even if the adversary obtains the public release, he will not be able to infer anything about specific individuals participating in this release.

For simplicity, we interpret databases as their histograms in $\NN_0^{|\mathcal{X}|}$, so we can focus only on unique rows and the numbers of their occurrences. 

\begin{definition}[Distance between databases]
The $\ell_1$ distance between two databases $x,y\in\NN_0^{|\mathcal{X}|}$ is defined as
$$
\|x-y\|_1 = \sum_{i\in \mathcal{X}}|x_i-y_i|,
$$
where $x_i$ and $y_i$ denote the numbers of occurrences of an item (an individual) $i$ in the databases $x$ and $y$, respectively. 
\end{definition}

One can easily see that $\|x-y\|_1$ measures how many records differ between $x$ and $y$. Moreover, $\|x\|_1$ measures the size of the database $x$.

A \textit{privacy mechanism} is a randomised algorithm used by the curator that takes a database as input and produces the output (the release) using randomisation.  

\begin{definition}[Differential Privacy --  from~\cite{DworkAlgo}]\label{dpDef}

A randomized algorithm $\mathcal{M}$ with domain $\mathbb{N}^{|\mathcal{X}|}$ is ($\varepsilon,\delta$)-differentially private (or ($\varepsilon,\delta$)-DP), if for all $\mathcal{S} \subseteq$ Range($\mathcal{M}$) and for all $x, y \in \mathbb{N}^{|\mathcal{X}|}$ such that $\left\Vert x-y \right\Vert_1 \leqslant 1$ the following condition is satisfied:
$$
\PR{\mathcal{M}(x) \in \mathcal{S}} \leqslant \mathrm{exp}(\varepsilon)\cdot \PR{\mathcal{M}(y) \in \mathcal{S}} + \delta,
$$
where the probability space is over the outcomes of the mechanism $\mathcal{M}$.
\end{definition}
When $\delta=0$, $\mathcal{M}$ is called ($\varepsilon$)-DP mechanism.

An intuition of $(\varepsilon,\delta)$-DP is as follows: if we choose two consecutive databases (that differ exactly on one record), the mechanism will likely return indistinguishable values. In other words, it preserves privacy with high probability, but it is admissible for a mechanism to be out of control with negligible probability~$\delta$.

\paragraph{Example 1} (Laplace noise)
In the central model, a standard and widely used mechanism with the $(\varepsilon)$-DP property is the so-called Laplace noise. A variable $X$ has Laplace distribution with parameter $\lambda$ (denoted as $X\sim\mathcal{L}(\lambda)$), if its probability density function is
$$
f(x)=\frac{1}{2\lambda}\exp\left(-\frac{|x|}{\lambda}\right)~.
$$
Note that  $\EE{X}=0$ and $\mathrm{Var}(X)=2\lambda^2$.

Let $c(x)$  be the number of rows in $x$, which satisfy a given property. Note that $c$ in the differential privacy literature is usually referred to as \textit{count query}.
Imagine that an aggregating mechanism is defined as follows: $\mathcal{M}(x)=c(x)+\mathcal{L}(\varepsilon^{-1})$. Then $\mathcal{M}$ is $(\varepsilon)$-DP (for more precise properties of Laplace noise, see \cite{DworkAlgo}).

In the privacy analysis of large-scale distributed protocols, two types of  approaches are typically distinguished: event/record-level (e.g., \cite{Wang_2019, papernot2017semisupervisedknowledgetransferdeep}), where information about a specific event/record is protected, and user-level privacy (e.g., \cite{mcmahan2018learningdifferentiallyprivaterecurrent}), which protects the privacy of the users themselves. Generally, the latter is a stronger model (see the discussion in \cite{mcmahan2018learningdifferentiallyprivaterecurrent}).

Fundamentally, the model analysed in our work pertains to event-level privacy. That is, we protect information regarding whether a counter incrementation event occurred. This corresponds to a scenario in which a participant can perform a certain action only once (for example, liking a specific message on typical social media platforms). These results can be directly adapted to scenarios where a single user may be associated with multiple counter increments using standard methods. Naturally, the strength of the privacy guarantee in such a framework must depend heavily on additional assumptions (such as constraints on the number of events or the associations between users).

 \section{Probabilistic Counters --- preliminaries}
\label{sect:counters}

This paper focuses on \textit{probabilistic counters}, further denoted by $M$. The notion of a probabilistic counter is ambiguous in the literature. It is a stochastic process that can be interpreted as a mechanism defined on the space of all possible inputs, which should estimate some goal value in some sense. The exact definition of this approximation is not crucial from the DP-point of view; thus we do not consider it in this paper.

Each increase in the data source counted by the probabilistic counter is called an \textit{increment request}. Due to the randomised nature of probabilistic counters, each may change the value of the counter, but not necessarily. We will also indicate the single increment request by $\true$. For the sake of generality, we also assume that the counter can get as an input $\false$, and in such a case it simply does nothing. This is useful for real-life scenarios, e.g., data aggregation (see Section~\ref{sect:scenario}). Obviously, only increment requests impact the counter's final result; hence, we indicate the counter's value after $n$ increment requests by $M_n$, and we are not considering the number of the rest of the rows.

In Figure~\ref{counterModel}, one can see a graphical representation of the probabilistic counter. As mentioned, increment requests are indicated by $\true$ and other rows by $\false$ input. The dice represent randomness. 
The $\mathbf{X}$-mark indicates that there is no action.

\begin{figure}[ht!]
    \centering
    \includegraphics[width=0.8\textwidth]{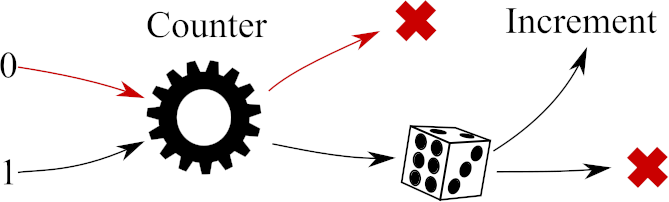}
    \caption{Graphical depiction of the probabilistic counter.}
    \label{counterModel}
\end{figure}

We emphasise that the probabilistic counter depends on the number of increment requests. We want to show that if we reveal its final value, then it does not expose any sensitive data about any single record. Moreover, note that if $x$ and $y$ differ only by one input $\false$, then $\PR{M(x) \in \mathcal{S}} = \PR{M_n \in \mathcal{S}} = \PR{M(y) \in \mathcal{S}}$, where $n$ is the number of increment requests for both $x$ and $y$. See that then the condition in Definition~\ref{dpDef} is trivially fulfilled. Hence, for our convenience, 
in this paper, we use the marking of only the number $n$ of increment requests provided by individuals when talking about the probabilistic counter $M_n$.

\begin{fact}\label{help}
Let $M$ be a probabilistic counter with a discrete $Range(M)=A$. Moreover assume that for all $n, m \gs 1$, such that $|n - m| \ls 1$, there exists such $S_n\subset A$ that for all $s\in S_n$
\begin{equation}\label{eq:counter1dp}
\PR{M_n=s}\ls \mathrm{exp}(\varepsilon)\cdot\PR{M_m=s}~
\end{equation}
and
\begin{equation}\label{eq:counterdelta}
\PR{M_n \notin S_n}\ls \delta~.
\end{equation}
Then $M$ is $(\varepsilon,\delta)$-DP.
\end{fact}

Note that, for our setting, Fact \ref{help} is fully compatible with the intuition of regular differential privacy (Definition~\ref{dpDef}).
Indeed, Fact \ref{help} can easily be derived from the observation that any set $B\subset A$ is a disjoint union of $B\cap S_n$ and $B\cap S_n'$.

Remark that $\varepsilon$ and $\delta$ in Fact \ref{help} can also be treated as functions of a parameter $n$, i.e. with respect to the number of increment requests. Thence, we can consider the differential privacy of the variable $M_n$, when $n$ is known, or ($\varepsilon(n),\delta(n))$-DP of the probabilistic counter $M$. The second variant lets us provide a precise dependence of privacy parameters of the counter as the number of increment requests gets large.

\subsection{Morris Counter}
\label{ssec:morris}

We begin with a short description of the Morris Counter (originally referred to as an approximate counter \cite{morris1978counting,flajolet1985approximate}). Fix $a>1$. The algorithm \ref{alg:morris} is a very simple pseudocode of the Morris Counter~\cite{morris1978counting}. 

\begin{algorithm}
	   $M \gets 1$\;
	   \While {receive request}{
		   generate $r\sim\mathrm{Uni}([0,1])$\;
	     \If {$r<a^{-M}$}{
            $M\gets M+1$\;
				}
			}
 \caption{Morris Counter Algorithm} 
 \label{alg:morris}
\end{algorithm}

Roughly speaking, we start with $M = 1$. Each incoming increment request triggers a random event. This event increments the counter ($M \leftarrow M + 1$) with probability $a^{-M}$ ($r\sim\mathrm{Uni}([0,1])$ generates a number uniformly at random from the interval $[0,1]$, using in practice some Pseudo Random Number Generator). Note that this approximate counting protocol can be easily distributed. Indeed, any entity who wants to increment the counter only has to send the request to increment it. These requests can be queued on the server and resolved one after another. A detailed description of the approximate counting method can be found in~\cite{morris1978counting,flajolet1985approximate}. Throughout this article, we examine only a standard Morris Counter  i.e., with the base $a=2$.
Morris Counter can also be defined recursively.
\begin{definition}\label{def:morris}
The Morris Counter is a Markov process $(M_n,n\in\NN_0)$ that satisfies:
\begin{align*}
&\PR{M_0=1}=1~,\\
&\PR{M_{n+1}=l|M_n=l}=1-2^{-l}~,\\
&\PR{M_{n+1}=l+1|M_n=l}=2^{-l}~,
\end{align*}
for any $l\in\NN$ and $n\in\NN_0$.
\end{definition}
Note that Definition \ref{def:morris} can be derived directly from a run of Algorithm~\ref{alg:morris}. From now on, let $\PR{M_n=l}=p_{n,l}$.
Directly from Definition \ref{def:morris} we get the following recursion:
\begin{equation}\label{eq:base}
p_{n+1,l}=(1-2^{-l})p_{n,l}+2^{-l+1}p_{n,l-1}~
\end{equation}
for $l\in\NN$ and $n\in\NN_0$ with starting and boundary conditions $p_{0,1}=1$, $p_{0,l}=0$ for $l\gs 2$ and $p_{n,0}=0$ for $n\in \NN_0$.

\paragraph{Accuracy versus Differential Privacy}

The accuracy of Morris Counter has been thoroughly analysed in various classical papers. The first detailed analysis was proposed by Ph.~Flajolet in \cite{flajolet1985approximate}.
In this part, we present the essence of theorems presented in this paper, which will be useful later on. 

First, we provide the asymptotics of the expected value and the variance of Morris Counter, with precise numerical approximations of constants:
\begin{fact}\label{fact:flajolet}
Let $M_n$ denote Morris Counter after $n$ successive increment requests. Then this random variable has an expected value $\EE{M_n} \approx \log(n) - 0.27395$ (in this paper $\log$ states for binary logarithm) and a variance $\mathrm{Var}(M_n) \approx 0.763014$.
\end{fact}

Realize that Fact \ref{fact:flajolet} guarantees high concentration of $M_n$ around its average --- a characteristic desirable in order to satisfy differential privacy definition.
Fact \ref{fact:flajolet} also justifies a definition of moving discrete intervals:
\begin{equation}
\label{eq:In}
I_n=[\left\lceil\log(n)\right\rceil-4:\left\lceil\log(n)\right\rceil+4]\cap [n+1]~,
\end{equation}
which will emerge as a crucial point of our further considerations of this Markov process in terms of differential privacy in this section. Let us mention that usually $[\left\lceil\log(n)\right\rceil-4:\left\lceil\log(n)\right\rceil+4]\subseteq [n+1]$, so one may think that $I_n$ are symmetric discrete intervals of length $8$, centred at $\lceil\log(n)\rceil$.

The lion's share of applications of Morris Counter is based on counting a number of occurrences, that is, the number of increment requests. In order to estimate this value, we may use ~(\ref{eq:base}) and simply obtain
$\EE{2^{M_{n+1}}}=\EE{2^{M_n}}+1$, so together with the assumption $M_0=1$ we obtain the following.
\begin{equation}
\label{eq:exp2}
\EE{2^{M_n}}=n+2~.
\end{equation}
Hence $2^{M_n}-2$ is an unbiased estimator of the number of increments $n$.
Remark that $n$ can be saved in $\lceil\log(n)\rceil$ bits. On the other hand, Fact \ref{fact:flajolet} shows that on average, $\log(\log(n))+O(1)$ bits are required to store $M_n$. As announced earlier, this is the crucial advantage of Morris Counter.
Moreover, analogously to (\ref{eq:exp2}) we may obtain 
\begin{equation}
\label{eq:morvar}
\mathrm{Var}(2^{M_n}-2)=\frac{n(n+1)}{2}~.
\end{equation}
Formulas (\ref{eq:exp2}) and (\ref{eq:morvar}) will be used in the example of data aggregation analysis in Section~\ref{sect:scenario}.

\subsection{MaxGeo Counter}

We begin with a short description of MaxGeo Counter. Algorithm \ref{alg:maxgeo} shows its pseudocode. Informally, for each increment request, the server has to generate a random variable from the geometric distribution $\mathrm{Geo}({1}/{2})$ (ranged in $\NN$). The final result is the maximum taken over all these random variables generated.

\begin{algorithm}
	   $C \gets 1$\;
	   \While {receive request}{
		   generate $r\sim\mathrm{Geo}({1}/{2})$\;
		   $C \gets \max\{C,r\}$\;
		   }
        return $C$\;
 \caption{MaxGeo Counter Algorithm} 
 \label{alg:maxgeo}
\end{algorithm}

The expectation and variance of the maximum of $n$ i.i.d. geometric variables have already been analysed in the literature. For instance, Szpankowski and Rego \cite{szpankowski1990yet} provided exact formulas for the' expected value and variance of such variables. However, they are impractical for large applications $n$. Hence, they also provided asymptotics (here, for a maximum of $n$ independent $\mathrm{Geo}(1/2)$ distributions):
$\EE{M_n}= \log(n)+ O(1)$ and 
$\mathrm{Var}(M_n)= \log(n)+ O(1)$ and thus, similarly to the Morris Counter, there are only $\log(\log(n)) +O(1)$ bits required on average to save the MaxGeo Counter after $n$ increment requests.



\subsection{General Probabilistic Counting with Stochastic Averaging}

Here we recall briefly a General Probabilistic Counting with Stochastic Averaging algorithm, based on the original idea from \cite{flajolet1985probabilistic}.
Assume that there are $m$, initially empty lots related to some independent copies of some probabilistic counter. For each increment request, we connect it to one of the groups uniformly at random. 
Finally, we perform incrementation requests  separately and independently for each lot, obtaining the following. 
$M[1],M[2],\ldots,M[m]$.

Without delving into details, for the original PCSA algorithm, $\EE{M_n}\approx\log(\varphi n)$, where $M_n$ is a value of a specific probabilistic counter connected with PCSA after $n$ increment requests and $\varphi$ is some magic constant.
If we denote the mean of these counters $m$ after the total number of increment requests $n$ by $\sigma_n(m)$, then we may introduce the statistic:
$$
\Xi_n(m)=\left\lfloor\frac{m}{\varphi}2^{\sigma_n(m)}\right\rfloor~.
$$
Then (according to \cite{flajolet1985probabilistic}), for any $m=2^k$, $k\in\NN$,
$
\EE{\Xi_n(m)}\approx n\left(1+\frac{0.31}{m}\right)
$
and
$
\mathrm{Var}(\Xi_n(m))= n^2 \left(\frac{0.61}{m}\right).
$

Note that averaging reduces the variance of the probabilistic counter.
Remark that ''Stochastic Averaging'' in PCSA algorithm refers to the random choice of the number of entities in each group, and it slightly differs from the standard averaging solution via the Monte Carlo method with groups of equal size.

An important conclusion is that we may apply the idea of original PCSA in general to any probabilistic counters.


\subsection{LogLog counter}

In~\cite{durand2003loglog} a LogLog algorithm was proposed.
It is based on $m=2^k$ counters $(M[j])_{j=1}^{m}$, where $k>0$. 
We may interpret this algorithm in the context of probabilistic counters. In such a scenario, it takes a hashed value (binary sequence) as input on every increment request.
The first $k$ bits of the hash determine which of the $m$ counters should be incremented (the index $j$ is chosen as the decimal representation of the sequence restricted to these first $k$ bits; hence it translates the increment request of the LogLog counter to the increment request of one of the $m$ internal counters).
Consider the first non-zero bit of a tail of the sequence (starting from $(k+1)$-th bit). Its position $R$ in this tail follows the $\mathrm{Geo}\left(\frac{1}{2}\right)$ distribution assuming the uniform distribution of the input sequences.
If $R>M[j]$, then $M[j]$ should become $R$. Otherwise, it does not change.

Therefore, LogLog counter is, in fact, a general PCSA that uses $m$ MaxGeo counters. It can be used to estimate the cardinality of increment requests $n$ using the following estimator: 
$$
\mathrm{LogLog}_n^{(m)}=\alpha_m m2^{\frac{1}{m}\sum\limits_{j=1}^{m}M_n[j]}~,
$$
with the scaling constant given by the formula
$$
\alpha_m=\left(\Gamma\left(-\frac{1}{m}\right)\frac{2^{-1/m} -1}{\ln(2)}\right)^{-m}~,
$$
where $\Gamma$ is Euler's Gamma function.

It is worth mentioning that $\alpha_m$ is an increasing sequence and $\alpha_m\approx 0.79$ for $m\gs~64$ (e.g. $\alpha_8\approx 0.69763\ldots$, $\alpha_{64}=0.78356\ldots$ and the limit is $\alpha_{\infty}=0.79402\ldots$).

We are interested in the expectation and accuracy of the cardinality estimator, which can be briefly described as follows:
$\EE{\mathrm{LogLog}_n^{(m)}}\approx n$ and $\mathrm{Var}(\mathrm{LogLog}_n^{(m)})\approx \frac{1.69 n^2}{m}$ (see \cite{durand2003loglog} for more details). However, in order to control the LogLog counter, on average about $m\log\log\left(\frac{n}{m}\right)$ bits of memory are needed.

An interesting fact is that the expectation of a single MaxGeo counter is logarithmic and the estimator of $n$ is of the form $Cm2^{\sigma_n(m)}$, for some constant $C$ (just as in the case of the original PCSA).

Since LogLog counter is an effect of processing of some MaxGeo counters, therefore its differential privacy is based on the same property of the auxiliary counters.

\subsection{HyperLogLog}

The maximum of geometric variables is also used as a primitive in the well known HyperLogLog algorithm (see~\cite{flajolet2007hyperloglog}). Therefore, its privacy properties are important both from the theoretical and practical point of view. Essentially, in HyperLogLog we perform the general PCSA algorithm, but the final estimation is somehow different:
$$
\mathrm{HyperLogLog}_n^{(m)} := \alpha_m m^2 \left(\sum_{j=1}^m 2^{-M_n[j]}\right)^{-1},
$$
where $\alpha_k$ is a constant dependent only on $k$ (see \cite{flajolet2007hyperloglog} for more details). It should be noted that HyperLogLog related algorithms (mentioned in Section \ref{sect:intro}) are the best-known procedures designated for cardinality estimation, and are close to optimum \cite{Indyk2003Woodruff}.
According to \cite{flajolet2007hyperloglog}, for $m=2^k$, where $k\gs 4$,
$$
\EE{\mathrm{HyperLogLog}_n^{(m)}}=n(1+\psi_3(n)+o(1)),\text{ with } |\psi_3(n)|<5\cdot 10^{-5}
$$
and
$$
\mathrm{Var(HyperLogLog)_n^{(m)}}=n^2\left(\frac{\beta_m}{\sqrt{m}}+\psi_4(n)+o(1)\right)^2, \text{ with } |\psi_4(n)|<5\cdot 10^{-4}~,
$$
where $\beta_m\stackrel{m\rightarrow\infty}{\longrightarrow} \sqrt{2\log(2)-1}=1.03896\ldots$ and $\beta_m\ls 1.106$ for $m\gs 16$.

 \section{Privacy-Preserving Survey via Probabilistic Counters}\label{sect:scenario}


In this section, we present an example scenario for data aggregation using probabilistic counters. We assume that there is a \textit{server} (alternatively, we call it \textit{aggregator}) and a collection of \textit{nodes} (e.g., mobile phone users), and we want to perform a boolean survey with a sensitive question. That is, each user sends $\false$ if his answer is \textit{no} and $\true$ if the answer is \textit{yes}. We assume that the connections between users and the server are perfectly secure, and the data can safely get to the trusted server. This can be performed using standard cryptographic solutions. The server's goal is to publish the sum of all responses $\true$ in a way that preserves privacy. This goal could obviously be achieved by simply collecting all the data and adding an appropriately calibrated Laplace noise (see~\cite{DworkAlgo}). However, we aim to show that probabilistic counters have inherently sufficient randomness to be differentially private without any auxiliary randomising mechanism.

We can present the general scenario in the following way:
\begin{enumerate}
 \item each user sends his/her bit of data to the server using secure channels,
 \item server plugs the data points sequentially into the counter,
 \item if the data point is $\true$, the counter receives \textit{increment request}, otherwise, the data is ignored,
 \item each increment request is being processed by the counter and may lead (depending on randomness) to an increase of the value of the counter,
 \item when all data are processed, the value of the counter is \textit{released} to the public.
\end{enumerate}
Note that we assume that the Adversary has access \textbf{only} to the released value. We also released only the counter value itself, which does not estimate the responses of $\true$. This estimation is a function of the released value, which is different for Morris or MaxGeo Counter. There can also be various ways to estimate the actual number using a counter value. However, this does not matter for our case, as differential privacy is conveniently fully resilient to post-processing (see~\cite{DworkAlgo}). The graphical representation of our scenario considered is presented in Figure~\ref{securityModel}.

\begin{figure}[ht!]
    \centering
    \includegraphics[width=0.8\textwidth]{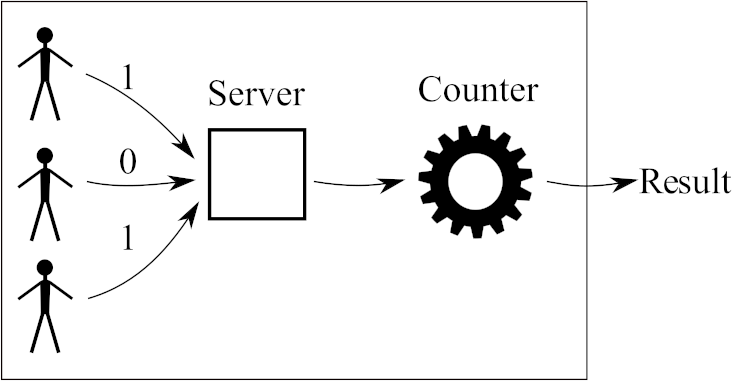}
    \caption{Scenario for data aggregation using probabilistic counters. We assume that the Adversary does not have any way to extract information from within the rectangle.}
    \label{securityModel}
\end{figure}

\paragraph{Adversary}
Our assumptions about the Adversary are the same as in most differential privacy papers. Namely, he may collude with any subset of the participants (e.g., all except the single user whose privacy he wants to breach). On the other hand, the aggregator is trusted. See that even though we have a distributed system in mind, this is, in fact, a central differential privacy scenario. We do not assume pan-privacy. This means that the algorithm's internal state is \textbf{not} subject to the constraints of differential privacy. Obviously, if the Adversary knew the internal state of the counter at any time or could observe whether, after receiving data from a specific user, the server had to perform computations to potentially increment the counter (implying a $\true$ response) or not, he would easily violate the privacy. We also do not assume privacy under continual observation. The survey is not published iteratively but once only after it is finished. In short, the Adversary \textbf{cannot} 
\begin{itemize}
 \item extract or tamper with the internal state of the counter,
 \item extracts any information from the server or channels between users and the server.
\end{itemize}
The Adversary \textbf{can}
\begin{itemize}
 \item collude with any subset $C$ of the participants (e.g., know their data or send them all $\false$ to the server) in order to breach the privacy of the user not belonging to $C$,
 \item obtain the final result of the aggregation and perform any desired post-processing on it.
\end{itemize}

 \section{Formulation of main results}
\label{results}

In this section, we state main results. The proofs are postponed to Section~\ref{sec:proofs} for the convenience of the reader.

\subsection{Morris Counter Privacy}
\label{morris}

In this subsection, we investigate the Morris Counter in terms of $(\varepsilon, \delta)$-DP. Here we present the main result:
\begin{theorem}\label{thm:main}
Let $M$ denote the Morris Counter and assume $|n-m|\ls 1$. Then 
$$
\PR{M_n=l} \leqslant \left(1-\frac{16}{n}\right)^{-1}\cdot\PR{M_m=l} + \delta,
$$
where $\delta<0.00033$,
so $M$ is $\left(L(n),0.00033\right)$-DP with
$$
L(n)=-\ln\left(1-\frac{16}{n}\right)=\frac{16}{n}+\frac{128}{n^2}+O(n^{-3})\ls \frac{16}{n-8}~.
$$
\end{theorem}

An explanation is postponed to subsection~\ref{proof1}.

\subsection{General  result on Morris' Counter privacy}

In this part, we show that the Morris' Counter guarantees privacy with both parameters tending fast to zero. The analysis is based on the observations from the previous case.  However, instead of $I_n$ (see Section~\ref{ssec:morris} for the discussion), we consider intervals 
$$
J_n(c)=[\lceil\log(n)\rceil - \lceil c\log(\ln(n))\rceil:\lceil\log(n)\rceil + \lceil c\log(\ln(n))\rceil]\cap[n+1]~,
$$
where $c$ is some positive constant such that $\lceil c\log(\ln(n))\rceil\gs 1$, for large enough $n$.

We can now state our next contribution:
\begin{theorem}\label{thm:addition}
Let $M$ denote the Morris Counter. If $c>0$ satisfies $\lceil c\log(\ln(n))\rceil\gs 1$, 
then 
$M$ is $\left(\varepsilon(n),\delta(n)\right)$-DP with parameters
$\varepsilon(n)=O\left(\frac{\left(\log(n)\right)^2}{n}\right)$ 
and 
$\delta(n)=O\left(n^{-\left(\ln(n)\right)^{c-1}} + n^{-1}\left(\ln(n)\right)^{-c}\right)$.
\end{theorem}

A proof is given in subsection~\ref{proof2}

\subsection{MaxGeo Counter Privacy}

In this subsection, we present a theorem that shows the privacy guarantees of MaxGeo Counter. Assume that we have $n$ increment requests. In the case of MaxGeo Counter, it means that we generate random variables $X_1,\ldots,X_n$, where  $X_i \sim \mathrm{Geo}({1}/{2})$ are pairwise independent. Ultimately, the result of the counter is maximum over all $X_i$'s, namely $X = \max(X_1,\ldots,X_n)$. Now we are ready to present our second main contribution.

\begin{theorem}\label{maxGeoTheorem}
Let $M$ denote the MaxGeo Counter, and $n$ denote the number of increment requests. Consider $m$ such that $|n-m| \leqslant 1$. Fix $\varepsilon > 0$ and $\delta \in (0, 1)$ and let 

\[
l_{\varepsilon} = \left\lceil \log\left(\frac{e^{\varepsilon}}{e^{\varepsilon}-1}\right) \right\rceil~.
\]
If
\begin{equation}
\label{eq:maxgeo}
n \gs \frac{\ln(\delta)}{\ln\left(1-2^{-l_{\varepsilon}}\right)} \left(\approx -\frac{\ln(\delta)}{\varepsilon}\right)~,
\end{equation}
then 
$$
\PR{M_n \in S} \ls \mathrm{exp}(\varepsilon)\cdot\PR{M_m \in S} + \delta,
$$
so $M$ is $(\varepsilon,\delta)$-DP.
\end{theorem}

A proof is postponed to subsection~\ref{proof3}.


 \section{Practical applications}
\label{practice}

\subsection{Discussion on arragements of the survey scenario}
We usually consider probabilistic counters to be some kind of estimator. There is always a trade-off between its precision and the quality of privacy.
Ultimately, it means that for fixed privacy parameters $(\varepsilon,\delta)$ we can calculate the minimum number of increment requests necessary to satisfy given privacy parameters. This can be done by artificially adding them before actually collecting data. Of course, it has to be taken into account that the initial added value should be subtracted from the final estimation of the appropriate cardinality before publication, and this change can impact the precision of the estimation (especially when the expected number of increment requests is very small). If we can perform such a preprocessing, then for every $(\varepsilon,\delta)$, we can easily know how many artificial counts have to be added. Nevertheless, we have to be aware that if the total number of increment requests is small, we may obtain a poor approximation, so this approach should be used whenever the privacy is much more important than the precision.


Note that, in light of our theorems~\ref{thm:main} and~\ref{maxGeoTheorem}, both the Morris Counter and the MaxGeo Counter preserve differential privacy in such a scenario. Assume that at least $n$ users have $\true$, therefore at least $n$ incrementation requests. See that we can either know it based on domain knowledge (e.g., we expect that at least some fraction of users will send $\true$ based on similar surveys) or add $x$ counts to the counter artificially initially. The number $x$ should be chosen according to the maximal amenable value of the parameter $\varepsilon$ for a given application, but we recommend choosing rather small values of $x$. Obviously, in the case of artificial counts, this has to be taken into account when estimating the final sum. Using the Morris counter, we obtain $\left(L(n),0.00033\right)$-DP with 
$$
L(n)=-\ln\left(1-\frac{16}{n}\right) \ls \frac{16}{n-8}~.
$$

\paragraph{Example 2}
Consider a result of Morris Counter with a small number $n$ of increment requests (for example, a number of respondents suffering from a rare illness). Therefore, we will likely require the parameter $\varepsilon$ to be at most some threshold, e.g. $1$. Therefore, from Theorem~\ref{thm:main}, we should add $x$ counts where $L(x) \ls {16}/{(x-8)} \ls 1$, so $x\geqslant 24$. Note that we do not include $n$ in the above formula since it is not known in advance. Therefore, using the Morris Counter, the above survey is at least $(L(n+24), 0.00033)$-DP. However, the estimator should be modified as well, i.e., $M'_n=M_{n+24}$, so $\hat{n'}=\max\{2^{M'_n}-26, 0\}$ (since $2^{M'_n}$ may be smaller than $26$).

On the other hand, using the MaxGeo Counter for a given $\varepsilon$ and $\delta$ we get $\left(\varepsilon, \delta\right)$-DP as long as $n \gs \dfrac{\ln(\delta)}{\ln\left(1-2^{-l_{\varepsilon}}\right)}$, where $l_{\varepsilon} = \left\lceil \log\left(1+{1}/{\varepsilon}\right) \right\rceil$ (see Theorem~\ref{maxGeoTheorem}). 
\paragraph{Example 3}
Assume that we have at least $n=200$ increment requests. From Theorem~\ref{thm:main}, we have $L(n) \ls {16}/{(n-8)} \ls 0.08334$. Hence, using the Morris Counter, the above survey is $(0.08334, 0.00033)$-DP.

On the other hand, using the MaxGeo Counter for a given $\varepsilon$ and $\delta$ we get $\left(\varepsilon, \delta\right)$-DP as long as $n \gs \dfrac{\ln(\delta)}{\ln\left(1-2^{-l_{\varepsilon}}\right)}$, where $l_{\varepsilon} = \left\lceil \log\left(1+{1}/{\varepsilon}\right) \right\rceil$. 

\paragraph{Example 4}
Let $\varepsilon = 0.5$ and $\delta = {1}/{D^2}$, where $D$ is the number of all the survey participants. After using our theorem and straightforward calculations, we have 
$
n \gs 7\ln(D)~.
$
Say we will have $\lfloor\exp(20)\rfloor$ participants. Then if we have at least $140$ increment requests, we satisfy $(0.5, {1}/{D^2})$-DP.

Note that from a differential privacy perspective, both the general PCSA algorithm and HyperLogLog can be seen as arbitrary postprocessing performed on the $m$ MaxGeo counters. Moreover, since each response goes to one counter only, they are independent of each other, so we can use the parallel composition theorem (see~\cite{DworkAlgo}). 

\begin{observation}
\label{obs:maxDP}
Assume we have $k$ MaxGeo Counters $M[1], \ldots, M[m]$, which are used either in HyperLogLog or PCSA algorithm. If $j$th MaxGeo Counter is $(\varepsilon_j, \delta_j)$-DP then the chosen algorithm is $(\max\limits_i \varepsilon_i, \max\limits_i \delta_i)$-DP.
\end{observation}

\subsection{Comparison of Morris and MaxGeo Counters}

In this subsection, we compare the' privacy and storage properties of data aggregation algorithms based on one of the investigated counters or the standard Laplace method.

We start with auxiliary remarks for the privacy of MaxGeo Counter.
For instance, see that if $\delta$ and $n$ are fixed, then from Theorem~\ref{maxGeoTheorem} and $l_{\varepsilon}\ls \left\lceil\ln\left(1+\varepsilon^{-1}\right)\right\rceil$ we obtain
\begin{equation}
\label{eq:epsmaxgeo}
\varepsilon(n)\gs \left(2^{\left\lfloor -\log\left(1 - \delta^{\frac{1}{n}}\right)\right\rfloor} - 1\right)^{-1}=:\varepsilon_0(n)~.
\end{equation}
We want to optimise $\varepsilon(n)$, so we will consider $\varepsilon_0(n)$ defined as the right-hand side of (\ref{eq:epsmaxgeo}).
In order to limit it let us consider the following function of $x\in\RR_+$: 
\begin{equation}
\label{eq:epsapprox}
\psi(x,\delta):=\left(\left(1 - \delta^{\frac{1}{x}}\right)^{-1} - 1\right)^{-1}=-\frac{\ln(\delta)}{x}+\frac{\ln(\delta)^2}{2 x^2} - \frac{\ln(\delta)^3}{6 x^3} +O(x^{-4})~.
\end{equation}
Naturally, then $\varepsilon_0(n)\gs\psi(n,\delta)=-{\ln(\delta)}/{n} + O(n^{-2})$. Since $\psi$ is decreasing with respect to $x$, we will consider when $\varepsilon_0(n)$ changes. More precisely,
consider a minimal $k$ such that $\varepsilon_0(n)< \psi(n-k,\delta)\ls \varepsilon_0(n-k)$, which appears to be the neat upper bound for $\varepsilon(n)$. However, since $\varepsilon_0(n)$ is the non-ascending step function, we realise that 
$$
\varepsilon_0(n-k)\gs \left(2^{\left\lfloor -\log\left(1 - \delta^{\frac{1}{n}}\right)\right\rfloor-1} - 1\right)^{-1}\!\!\!\!=-\frac{2\ln(\delta)}{n}+\frac{3\ln(\delta)^2}{n^2}-\frac{13\ln(\delta)^2}{3n^3}+O(n^{-4})~.
$$
If we denote $\phi(n,\delta):=\left(2^{\left\lfloor -\log\left(1 - \delta^{\frac{1}{n}}\right)\right\rfloor-1} - 1\right)^{-1}$, then we can summarise our recent considerations in a short way by $\psi(n,\delta)\ls\varepsilon(n)<\phi(n,\delta)$.
Thence, in the case where we fix the parameter $\delta=0.00033$, we obtain 
$$
\frac{8.0164\ldots}{n}+\frac{32.13147\ldots}{n^2} + O(n^{-3})\ls\varepsilon(n)\ls\frac{16.0328\ldots}{n} + \frac{192.789\ldots}{n^2}+O(n^{-3})~.
$$
On the other hand, from Theorem~\ref{thm:main}, we know that when $\delta=0.00033$, then for Morris Counter (with $\varepsilon(n)$ defined by (\ref{eq:epsilon})) the quite similar relation holds:
$$
\varepsilon(n)\ls-\ln\left(1-\frac{16}{n}\right)=\frac{16}{n}+\frac{128}{n^2}+O(n^{-3})~.
$$
Therefore, Morris and MaxGeo Counters behave quite similarly under comparable conditions, and Figure \ref{fig:comparison} confirms this observation. Indeed, in Figure \ref{fig:comparison} we can see that the difference between the values of the parameters $\varepsilon(n)$ for both counters decreases as $n$ increases.

\begin{figure}[ht!]
    \centering

    \includegraphics[width=0.95\textwidth]{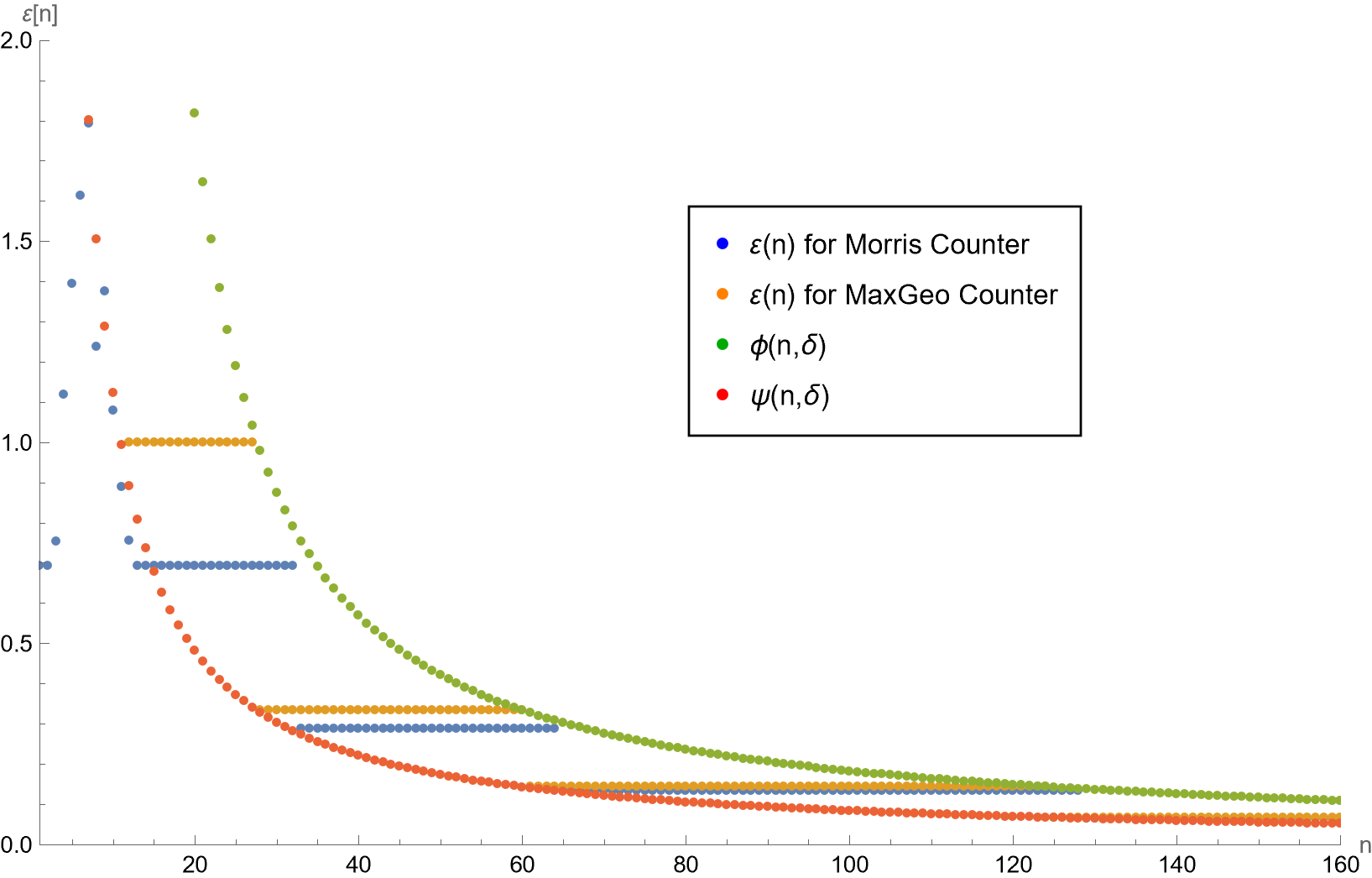}
    \caption{Values of $\varepsilon(n)$ parameters for Morris and MaxGeo Counters compared with boundaries for $\varepsilon(n)$ for MaxGeo Counter: the lower one --- $\psi(n,\delta)$ and the upper one --- $\phi(n,\delta)$ ($n\ls 160$ and $\delta=0.00033$).}
    {\label{fig:comparison}}
\end{figure}

Realise that the previous conclusions remain true if $\delta(n)$ is not constant. This short observation enables us to obtain a more general result:
\begin{fact}
Let $\delta(n)=n^{-c}$ for some constant $c>0$. Then 
$$
\varepsilon(n)\ls\phi(n,\delta(n))=\frac{2c\ln(n)}{n}+\frac{3c^2\ln(n)^2}{n^2}+O\left(\frac{c^3\ln(n)^3}{n^3}\right)
$$
and MaxGeo Counter is $(\phi(n,\delta(n)),\delta(n))$-DP for any $n\in\NN$.
\end{fact}

Notice that in this case, both sequences of parameters tend to $0$, which may be used as an advantage in applications, especially when we expect that the total number of increment requests will be very large.
However, we emphasise that this requires that $\delta(n)$ be negligible.

\begin{figure}[h]
\centering
\begin{tabular}{|l*{3}{|c}|}
\hline
Method & Laplace noise & Morris counter & $\mathrm{LogLog}^{(2)}$ counter \\
 & $\mathcal{L}(n/16)$ & $M_n$ & $(M_n[1],M_n[2])$ \\
\hline
\rule{0mm}{4mm}$(\varepsilon,\delta)$-DP & $(16/n,0)$-DP & $ (\frac{16}{n-8},\delta)$-DP & $(\sim\frac{32.066}{n},\delta)$-DP \\[1mm]
\hline
\rule{0mm}{5mm}Estimator $\hat{n}$ & $n +\mathcal{L}(n/16)$ & $2^{M_n} -2$ & $\sim0.89\cdot2^{\frac{1}{2}\sum_{i=1}^{2} M_n[i]}$ \\[2mm]
\hline
\rule{0mm}{4mm}$\mathrm{Var}(\hat{n})$ & $\frac{n^2}{128}$ & $\frac{n^2+n}{2}$ & $\sim 0.845\cdot n^2$ \\[1mm]
\hline
\rule{0mm}{4mm}Avg. memory & $\log(n)+O(1)$ & $\log(\log(n)) + O(1)$ & $2\log(\log(\frac{n}{2}))+O(1)$ \\[1mm]
\hline
\end{tabular}
\caption{A summary of data aggregation techniques. The standard one is based on the Laplace method, and the rest are based on probabilistic counters. Recall that $\delta=0.00033$.}
\label{fig:comparison2}
\end{figure}

In Figure \ref{fig:comparison2}, we may briefly see that probabilistic counters can be used for data aggregation to decrease memory usage in exchange for a slight increase of the parameter $\delta$ of differential privacy and wider confidence intervals (lower accuracy).
In recent years Big Data related problems became very popular. Note that this kind of application makes major use of memory. When the server aggregates many different data, standard solutions may cause a serious problem with data storage, which can be encountered by using the idea based on a probabilistic counter instead.
\paragraph{Example 5}
Imagine that 100 million people participate in a general health survey with $100$ \textit{yes/no} sensitive questions.For every question, we would like to estimate the number of people who answered \textit{yes}, but we want to guarantee the differential privacy property at a reasonable level. Realise that if the number of \textit{yes} answers is very small for some questions (e.g., when the question is about a very rare disease), then the number of \textit{no} answers may be counted instead. Obviously, this method gains privacy, but loses precision.

According to Figure \ref{fig:comparison2}, if we use the Laplace method, then we may need approximately $100\log(10^8)=2657.54\ldots$ bits to store the counters. Let us note however, if we use the Morris Counter instead, about $100\log(\log(10^8))=473.20\ldots$ bits are needed. Note that all terms $O(1)$ in the ''Average memory'' row of Figure \ref{fig:comparison2}
are bounded by $1$. Hence, its impact is negligible from a practical point of view.

One may also complain about the heavy use of the pseudo random number generator (PRNG) that probabilistic counters make. However, this problem may be resolved by generating the number of increment requests which have to be forgotten until the next update of the counter by using appropriate geometric distributions (see, for example, \cite{10.1145/198429.198435} for a similar approach applied to the reservoir sampling algorithm). This way, the use of PRNG can be substantially reduced.

 \section{Proofs}
\label{sec:proofs}

\subsection{Main Theorem for Morris Counter (Theorem \ref{thm:main})}
\label{proof1}

The proof is complicated and very technical. In order to better understand it, we are going to provide a presentation of a plan and main ideas beneath the parts of the proof. Let $n_k:=2^k+1$ for $k\in\NN$. 

We introduce $\mathcal{P}_k^{(c)}$ as the probability that the Morris Counter $M_n$ being $k+c$ after $n_k$ increment requests, i.e., $p_{n_k,k+c}$. The ''special'' sequences $(\mathcal{P}_k^{(c)})_k$ play a crucial role in the proof. 
\paragraph{Roadmap of the proof:}
We can divide the proof of Theorem \ref{thm:main} into five phases (the main results of the phases are given in brackets):
\begin{enumerate}
\item $\delta$ phase 
(Theorem \ref{thm:morrisdelta}),
\item relations between ''special'' sequences $(\mathcal{P}_{k}^{(c)})_k$ with respect to $c$ (Claim \ref{claim4}),
\item dependencies between consecutive distributions (of $M_{n}$ and $M_{n+1}$) (Claim \ref{stepforward}),
\item extrapolation of $\mathcal{P}_k^{(c)}\leq 2^{c+3} \mathcal{P}_k^{(c+1)}$ property to $N>n_k$ (Lemma \ref{lem:crutial}),
\item $\varepsilon$ phase (Theorem \ref{thm:morriseps}).
\end{enumerate}

During the first phase, we consider a concentration of Morris Counter in the vicinity of its mean value. 
More precisely, we show that the Morris Counter after $n$ increment requests takes values in relatively small intervals $I_n$ with probability $1-\delta$ (note that then $M_n$ satisfies condition (\ref{eq:counterdelta}) for $S_n=I_n$), where $I_n$ is defined as in (\ref{eq:In}) and $\delta$ is some small constant, which arises from the proof.
Note that $I_n$ can be interpreted as confidence intervals at level $1-\delta$ (see, e.g.~\cite{NeymanEstimation37}).
This phase is divided into lemmas \ref{lem:lowerdelta} and \ref{lem:upperdelta}.
The first one uniformly bounds the formula for probabilities given
by the Theorem below, due to Flajolet:
\begin{theorem}[Proposition 1 from~\cite{flajolet1985approximate}]\label{flajoletThm}
The probability $p_{n,l}$ that the Morris Counter has value $l$ after $n$ increment requests is
$$
p_{n,l} = \sum_{j=0}^{l-1} \left(-1\right)^j 2^{-j \left( j-1 \right)/2} \left(1-2^{-\left(l-j\right)}\right)^n \prod_{i=1}^{j} \left(1-2^{-i}\right)^{-1} \prod_{i=1}^{l-1-j} \left(1-2^{-i}\right)^{-1}~.
$$
\end{theorem}

We sum up the bounds on $p_{n,l}$ to obtain a small upper bound for $\delta_1:=\PR{M_n \ls \left\lceil\log(n)\right\rceil-5}$ (Lemma \ref{lem:lowerdelta}).
The same bounds cannot be utilised efficiently in the proof of Lemma \ref{lem:upperdelta}. 
Instead, it couples $M_n$ with a process $X_n$, which increases during the first $\left\lceil\log(n)\right\rceil +1$ steps, and then follows the same update rule, so $M_n\ls X_n$ almost surely. 

Therefore $\delta_2:=\PR{M_n \gs \left\lceil\log(n)\right\rceil+5} \ls \PR{X_n \gs \left\lceil\log(n)\right\rceil+5}$, which is much easier to bound from definition. Note that such a coupling cannot be used in the proof of Lemma \ref{lem:lowerdelta}. The first phase is summarised by Theorem \ref{thm:morrisdelta}, i.e., establishes $\delta=\delta_1+\delta_2$.

In the second phase, we show that   $(\mathcal{P}_k^{(4)})_{k}$ is descending for large enough $k$ and $(\mathcal{P}_k^{(5)})_{k}$ is ascending for big enough $k$ (Lemma \ref{lem:main}). A change of monotonicity is the first obstacle in the proof of the main theorem.
The main idea beneath the proof of this fact is to calculate the differences between consecutive elements of the considered ''special'' sequences by representing them as the sums via application of Flajolet's Theorem \ref{flajoletThm} and realising that usually at most first ten terms of the sums are crucial (on the other hand, taking less than eight terms is rarely sufficient). This is the second issue, which makes the proof so complicated.
Let us note that Theorem \ref{flajoletThm} presents an explicit formula for $\PR{M_n=l}$, which (as we may experience in \ref{append}) is not convenient to analyse. 
However, it is simple enough to find the values numerically (also note that recursive Definition \ref{def:morris} provides those probabilities easily as well, However, this approach is inefficient in terms of memory and time for a large number of requests $n$). 
Therefore, by precise analysis, we can finally check some sums numerically and obtain the thesis of Lemma \ref{lem:main} for $k\gs 15$. However, numerically, one can extrapolate it to some smaller $k$ as well (remark that this proof does not work for small values of $k$, since not only ten terms of the aforementioned sums are important).
The lemma \ref{lem:main} can be used directly to show that $\mathcal{P}_k^{(4)}\ls 2^7 \mathcal{P}_k^{(5)}$ for $k\gs 7$ (Claim~\ref{claim4}). Let us note that, in fact, this result is not true for $k<7$.

The third phase is based on the foregoing intuition: If the ratios of consecutive probabilities of distribution of $M_n$ increase almost exponentially, then the ratios of distribution of $M_{n+1}$ increase similarly (but slightly slower). The proofs mainly use the definition of the Morris counter.
{}

The fourth phase (i.e. Lemma \ref{lem:crutial})
shows that if $p_{n_k,k+c}\ls 2^{c+3} p_{n_k,k+c+1}$, for $c\in[-k:4]$, then the same is true if we substitute $n_k$ with a bigger number (i.e., this property is increasing with respect to the $n$ parameter). 
In order to apply this result, we have to satisfy some starting conditions. We have numerically checked the appropriate condition for $k=7$ (presented later, in Table \ref{tab:alfa}). Therefore, the second phase of the proof of Theorem \ref{thm:main} justifies the assumptions of Lemma \ref{lem:crutial} with respect to the parameter $k$, as long as $k\gs 7$ and the third phase let us obtain the appropriate assumptions with respect to $c$. One can check that for $k<7$, an analogous assumption is not true.

The latter phase begins with the application of the result from the previous part. 
We obtain $\varepsilon$ 
is at most $L(n)$ (provided in the formulation of Theorem \ref{thm:main}), for $k\gs 7$. 
The last piece of this puzzle is justified by a numerical evaluation for $k<7$ (presented later, in Figure \ref{fig:epsilon}), which ends the $\varepsilon$ phase and so the whole proof.




\paragraph{Proof}

In the following, we present the proof of the main contribution. Nevertheless, some technical lemmas are given in the \ref{append}.

\paragraph{$\delta$ phase}

Let us begin with a reminder. First, $M_n \in \NN_0$ and, moreover,
\begin{equation*}
I_n \subseteq 
\left[\left\lceil\log(n)\right\rceil - 4 :
      \left\lceil\log(n)\right\rceil + 4\right].
\end{equation*}

We provide few facts about the concentration of the distribution of the random variable $M_n$, or more precisely about the probability that $M_n$ will be outside the interval $I_n$.
\begin{lemma}
\label{lem:lowerdelta}
Let $M_n$ be the state of the Morris counter after the $n$ increment requests. Then
$$
\delta_1:=\PR{M_n \ls \left\lceil\log(n)\right\rceil-5}\ls 0.000006515315\ldots~.
$$
\end{lemma}

An increasing sequence $\prod\limits_{i=1}^{k} \left(1-2^{-i}\right)^{-1}$ that emerged in Theorem \ref{flajoletThm} will be indicated by $r_k$ (with $r_0=1$) and we denote its limit $\prod\limits_{i=1}^{\infty} \left(1-2^{-i}\right)^{-1}= 3.46274\ldots$ by $R$.\\ Let us mention this with the notions $q_k$ and $Q$ from  \cite{flajolet1985approximate}, $r_k=\frac{1}{q_k}$ and $R=\frac{1}{Q}$.

\begin{proof}
At first, we want to bound a lower tail of the distribution $\delta_1$.

Here we would like to find a sufficient upper limit for the above probability.
Assume that $l \ls \left\lceil\log(n)\right\rceil-5$.

 Realise that $r_k\ls R$ and that $y=-x+1$ is a tangent line to the plot of $y=-x(x-1)$ in the point $(x,y)=(1,0)$.
Therefore:
\begin{align*}
p_{n,l}\stackrel{\mathrm{Thm}~\ref{flajoletThm}}{\ls} \sum_{j=0}^{l-1} 2^{-\frac{j \left( j-1 \right)}{2}} \left(1-2^{-\left(l-j\right)}\right)^n r_j r_{l-1-j}
\ls R^2 \left(1-2^{-l}\right)^n \sum_{j=0}^{l-1} \sqrt{2}^{-j+1}\\
\ls R^2 \frac{2}{\sqrt{2}-1} \exp(-n2^{-l})=R^2 (2\sqrt{2}+2) \exp(-n2^{-l}).
\end{align*}

The above formula will help us limit the left tail of the distribution of $M_n$:

\begin{multline*}
\delta_1
=\sum_{l=1}^{\left\lceil\log(n)\right\rceil-5}\PR{M_n=l}
 \ls R^2 (2\sqrt{2}+2)\sum_{l=1}^{\left\lceil\log(n)\right\rceil-5} \exp(-n2^{-l}) \\
\ls R^2 (2\sqrt{2}+2) \sum_{k=4}^{\infty}\exp(-2^k)
\ls R^2 (2\sqrt{2}+2) \sum_{k=1}^{\infty}\exp(-16k)\\
=R^2 (2\sqrt{2}+2)\frac{\exp(-16)}{1-\exp(-16)}= 0.000006515315\ldots~.
\end{multline*}
\end{proof}

\paragraph*{Remark}
The bound for $p_{n,l}$ obtained above is useless when $l\gs \log(n)-2$, so it cannot be used in the next lemma for a symmetric upper tail.

\begin{lemma}
\label{lem:upperdelta}
Let $M_n$ be the state of the Morris Counter after the $n$ increment requests. Then 
$$
\delta_2:=\PR{M_n \gs \left\lceil\log(n)\right\rceil+5}\ls 0.000325521\ldots~.
$$
\end{lemma}

\begin{proof}

Consider a process $X=(X_{k\in\! [0:n]})$. Let $X$ initially follow the incrementation rule $\PR{X_k=k+1}=1$ for $k\in[0:\lceil\log n \rceil+1]$. Afterwards, let this Markov chain imitate the transition rule of Morris Counter, that is 
$$
\PR{X_{k+1}=m+1|X_k=m}=\frac{1}{2^m}=1-\PR{X_{k+1}=m|X_k=m}
$$
for $k\gs \left\lceil \log(n)\right\rceil +1$.
Naturally, for $k\ls \left\lceil \log(n)\right\rceil +1$, we have $X_k \gs M_k$, so we may couple realisations of these two processes in such a way that whenever $X$ increases, so is $M$ and if $M$ does not change, then $X$ does not increase as well (note that $X$ has at most the same probability of a positive increase as $M$ at any point in time).

To abbreviate the expressions, let us denote $m=n-\left\lceil\log(n)\right\rceil - 1$ and
$$
\mu_\iota=\PR{X_{k+1}=\left\lceil\log(n)\right\rceil+\iota+1 | X_k =\left\lceil\log(n)\right\rceil+\iota}=\frac{1}{2^{\left\lceil\log(n)\right\rceil+\iota}}=1-\nu_\iota~,
$$
for any $\iota\in\ZZ$.
Moreover, let us define a three-dimensional discrete simplex:
$$
S_k^{(3)}=\{\bar{l}=(l_1,l_2,l_3)\in\NN_0^3:l_1+l_2+l_3\ls k\}~.
$$


Thus,

\begin{align*}
\delta_2
\ls \PR{X_n\gs \left\lceil \log(n)\right\rceil +5}
&=\sum_{\bar{l}\in S_{m-3}^{(3)}} \nu_2^{l_1} \mu_2 \nu_3^{l_2} \mu_3 \nu_4^{l_3} \mu_4
\ls \sum_{\bar{l}\in S_{m-3}^{(3)}} \dfrac{1}{2^{3\left\lceil\log(n)\right\rceil +9}}\\
&=\sum_{k=0}^{m-3} \binom{k+3}{2} \dfrac{1}{2^{3\left\lceil\log(n)\right\rceil +9}}
\ls \frac{1}{2^{10} n^3}\sum_{k=3}^{m}k^2-k~.
\end{align*}

Realise that $\sum\limits_{k=3}^{m}k=(m-2)(m+3)/2$ and $\sum\limits_{k=3}^{m}k^2 = (m-2)(2m^2+7m+15)/6$, so
\begin{align*}
\delta_2&
\ls \frac{1}{2^{10} n^3}\frac{1}{6}(m-2)(2m^2+4m+6)=\frac{1}{3\cdot2^{10} n^3}(m^3-m-6)\\
&\ls \frac{m^3}{3\cdot2^{10} n^3}\ls \frac{1}{3\cdot2^{10}}=0.000325521\ldots~.
\end{align*}
Note that when $m<3$ (that is, when $n<7$), then the above sums are empty, but on the other hand $\lceil\log(n)\rceil+5>n+1$, so the inequality is trivially true.
\end{proof}

\begin{theorem}\label{thm:morrisdelta}
The state of the Morris Counter after $n$ increment requests is \textbf{not} in the set 
$$
I_n=[\left\lceil\log(n)\right\rceil-4: \left\lceil\log(n)\right\rceil+4]\cap [n+1]~
$$
with probability $\delta < 0.00033$.
\end{theorem}

\begin{proof}
Realise that $\PR{M_n\in [1:n+1]}=1$. 
This observation, together with the lemmas \ref{lem:lowerdelta} and \ref{lem:upperdelta} yields 
$$
\delta := \PR{M_n\notin I_n}=\delta_1 + \delta_2 
< 0.00033~.
$$

\end{proof}

\paragraph{The second phase}
In this part of the investigation, we try to establish the $\varepsilon(n)$ parameter of DP of $M_n$.
In fact, it remains to examine the property (\ref{eq:counter1dp}) in the interval $I_n$, as Theorem \ref{thm:morrisdelta} entails (\ref{eq:counterdelta}) for $S_n=I_n$.
Therefore, we are interested in finding the upper bound for maximal privacy loss for any $n\in\NN$ and $k\in I_n$, namely:
\begin{equation}
\label{eq:epsilon}
\varepsilon(n)=\max\left\{\left|\ln\left(\frac{p_{n\pm 1,k}}{p_{n,k}}\right)\right|:k\in I_n\right\}~.
\end{equation}
Actually, we may consider the sign $'+'$ instead of $'\pm'$ in (\ref{eq:epsilon}), because $|\ln(x)|=|\ln({1}/{x})|$.
However, when $I_n\neq I_{n\pm 1}$, we have to behave carefully, so in particular, an additional cheque of privacy loss with the sign $'-'$ is needed when $n$ is of a form $2^l+1$ for some $l\in\NN$.

\begin{claim}
\label{claim4}
For $k\gs 7$, we have $p_{2^k+1,k+4}\ls 2^7 p_{2^k+1,k+5}$.
\end{claim}
The above claim is the result of a simple application of Lemma \ref{lem:main} from \ref{append}.

\paragraph{The third phase}

\begin{claim}
\label{stepforward}
If for any given $n$, there exists an ascending and positive sequence $(\alpha_i)_{i=1}^{n}$ such that $$(\forall\; i\in [1:n])\; p_{n,i}= 2^i \alpha_i p_{n,i+1},$$ then there also
exists an ascending and positive sequence $(\alpha_i')_{i=1}^{n+1}$ such that 
$$
(\forall\; i\in [1:n+1])\; (p_{n+1,i}= 2^i \alpha_i' p_{n+1,i+1})\wedge (\forall\; i\in [1:n])\; (\alpha_i'<\alpha_i)~.
$$ 
\end{claim}
This claim arises from lemmas \ref{lem:between} and \ref{lem:end} from \ref{append}. 

\paragraph{The fourth phase}
We use Claim \ref{stepforward} to guarantee starting conditions for the next Lemma \ref{lem:crutial}. However, in order to apply Lemma \ref{lem:crutial}, we will also use Claim \ref{claim4}, which assumes that $n\gs 2^7 +1$. Hence, we would like to gather some information about the distribution of $M_{2^7+1}$. More precisely, we are interested in the behaviour of $\theta_i={p_{129,i}}/{p_{129,i+1}}$ for $i\ls 11$, which we present in Table \ref{tab:alfa}.
\begin{table}[h!]
\centering
\begin{tabular}{c|*{3}{c}}
i & $\theta_i$ & $2^{i-4}$ & $2^{4-i} \theta_i$\\
\hline\hline
1 & $9.6205\ldots\cdot 10^{-24}$ & 0.125 & $7.6964\ldots\cdot 10^{-23}$\\
\hline
2 & $1.73351\ldots\cdot 10^{-9}$ & 0.25 & $6.93402\ldots \cdot 10^{-9}$\\
\hline
3 & $0.000119359\ldots$ & 0.5 & $0.000238718\ldots$\\
\hline
4 & $0.0140238\ldots$ & 1 & $0.0140238\ldots$ \\
\hline
5 & $0.158163\ldots$ & 2 & $0.0790814\ldots$ \\
\hline
6 & $0.771817\ldots$ & 4 & $0.192954\ldots$ \\
\hline
7 & $2.67702\ldots$ & 8 & $0.334628\ldots$ \\
\hline
8 & $7.83367\ldots$ & 16 & $0.489604\ldots$ \\
\hline
9 & $20.8095\ldots$ & 32 & $0.650297\ldots$ \\
\hline
10 & $52.0472\ldots$ & 64 & $0.813238\ldots$ \\
\hline
11 & $125.065\ldots$ & 128 & $0.977073\ldots$ \\
\end{tabular}
\caption{Ratios of adjacent probabilities of the distribution of $M_{2^7+1}$, compared with the exponential function of base $2$.}
\label{tab:alfa}
\end{table}
We briefly see a superexponential trend of proportions $\theta_i$, so the possibility of using Claim \ref{stepforward} for $n\gs 2^7+1$ is justified. It might seem that the choice of $n$ is arbitrary, but it occurs that the distribution of $M_{2^6+1}$ does not satisfy the necessary assumptions for privacy loss, although $M_{2^6+1}$ can still fulfil the property of $(\varepsilon(n),\delta)$-DP with the parameters given in Theorem \ref{thm:main}.
\begin{lemma}
\label{lem:crutial}
Let $k\in\NN\setminus\{0\}$ and $n_k=2^{k}+1$. If $p_{n_k,k+c}\ls 2^{c+3}p_{n_k,k+c+1}$ for every $c$ in the interval $[-k:4]$, then 
$$
(\forall\; N\gs n_k)(\forall\; c\in [-k:4])\; p_{N,k+c}< 2^{c+3}p_{N,k+c+1}~.
$$
\end{lemma}

\begin{proof}
Realise that for $c=-k$, the required inequality is trivial.
Therefore, we can safely consider only $c\in [-k+1:4]$. We would like to prove this inductively with respect to $c$ and $N$. Assume that for some $N\gs n_k$ and any $d\in\{c-1, c\}$ we have $p_{N,k+d}< 2^{d+3}p_{N,k+d+1}$. 

Then also
\begin{align*}
p_{N+1,k+c}&\eqthr p_{N,k+c}(1-2^{-k-c})+p_{N,k+c-1} 2^{-k-c+1}\\
&\ls 2^{3+c} p_{N,k+c+1}(1-2^{-k-c})+2^{3+(c-1)} p_{N,k+c} 2^{-k-c+1}\\
&< 2^{3+c} (p_{N,k+c+1} (1-2^{-k-c-1})+2^{-k-c} p_{N,k+c})=2^{3+c}p_{N+1,k+c+1}~.
\end{align*}
If we start with $c=-k+1$, then let us prove inductively the appropriate condition for all $N\gs n_k$.
In addition, the thesis is followed by the induction with respect to $c$.
\end{proof}

\paragraph{$\varepsilon$ phase}
Claims \ref{claim4} and \ref{stepforward}, together with Table \ref{tab:alfa} enable us to apply Lemma \ref{lem:crutial} for $n=2^k+1$ for any $k\gs 7$.

\begin{theorem}\label{thm:morriseps}
Let $n> 2^7=128$ and $k\in I_n$. Then 
$$
1-\frac{16}{n}\ls \frac{p_{n\pm 1,k}}{p_{n,k}}\ls 1+\frac{16}{n}~.
$$
\end{theorem}

\begin{proof}
According to the previous discussion on formula (\ref{eq:epsilon}), we examine 
$$
\frac{p_{n+1,k}}{p_{n,k}}\eqthr \frac{p_{n,k}(1-2^{-k})+2^{-k+1}p_{n,k-1}}{p_{n,k}}=1+2^{-k}\left(-1+2\frac{p_{n,k-1}}{p_{n,k}}\right)~.
$$
Let us denote $l=\left\lceil\log(n)\right\rceil$ and $c=k-l\in [-4:4]$.
Then Lemma \ref{lem:crutial} bears $p_{n,k-1}\ls~2^{c+3}p_{n,k}$, so
$$
\frac{p_{n+1,k}}{p_{n,k}}\ls 1+2^{-l-c}(-1+2^{c+4})<1+2^{-l+4}=1+\frac{16}{2^{\lceil\log(n)\rceil}}<1+\frac{16}{n}~.
$$
Realise that if $n=2^{l-1}+1$ for some $l\in\NN$, then a little adjustment is necessary. Indeed, let now $c-1=k-l\in[-4:4]$, and once again, Lemma \ref{lem:crutial} provides $p_{n-1,k-1}<~2^{c+2}p_{n-1,k}$. However, it still holds that:
$$
\frac{p_{n,k}}{p_{n-1,k}}=1+2^{-k}\left(-1+2\frac{p_{n-1,k-1}}{p_{n-1,k}}\right)\ls 1+2^{-l-c+1}\left(-1+2^{c+3}\right)< 1+\frac{16}{n}~.
$$
On the other hand, we  have inequalities $p_{n+1,k}>\left(1-2^{-l-c}\right)p_{n,k}$ and $p_{n,k}>\left(1-2^{-l-c}\right)p_{n-1,k}$ for any $c\in [-4:4]$, so both fractions exceed $1-{16}/{n}$.
\end{proof}

Theorem \ref{thm:morriseps} only provides $\varepsilon(n)\ls -\ln\left(1-{16}/{n}\right)$ for $n>128$ (compare with (\ref{eq:epsilon})). However, in Figure \ref{fig:epsilon} we may briefly see that the above inequality is true for smaller numbers of requests $n$ as well.

Having all the technical lemmas, we are now ready to prove Theorem~\ref{thm:main}.


\begin{proof}(of Theorem~\ref{thm:main}) 
Suppose that $S_n=I_n$ in Fact \ref{help}. Then from theorems \ref{thm:morrisdelta} and \ref{thm:morriseps} we can easily see that $\PR{M_n\notin S_n}<0.00033$ and 
$$
(\forall\; m\in \{n-1,n+1\})\;(\forall\; l\in S_n)\; \PR{M_n=l} \leqslant \left(1-\frac{16}{n}\right)^{-1}\cdot\PR{M_m=l}~,
$$
hence, from Fact \ref{help} we obtain the main result.
\end{proof}

\begin{figure}[ht!]
    \centering
    \includegraphics[width=\textwidth]{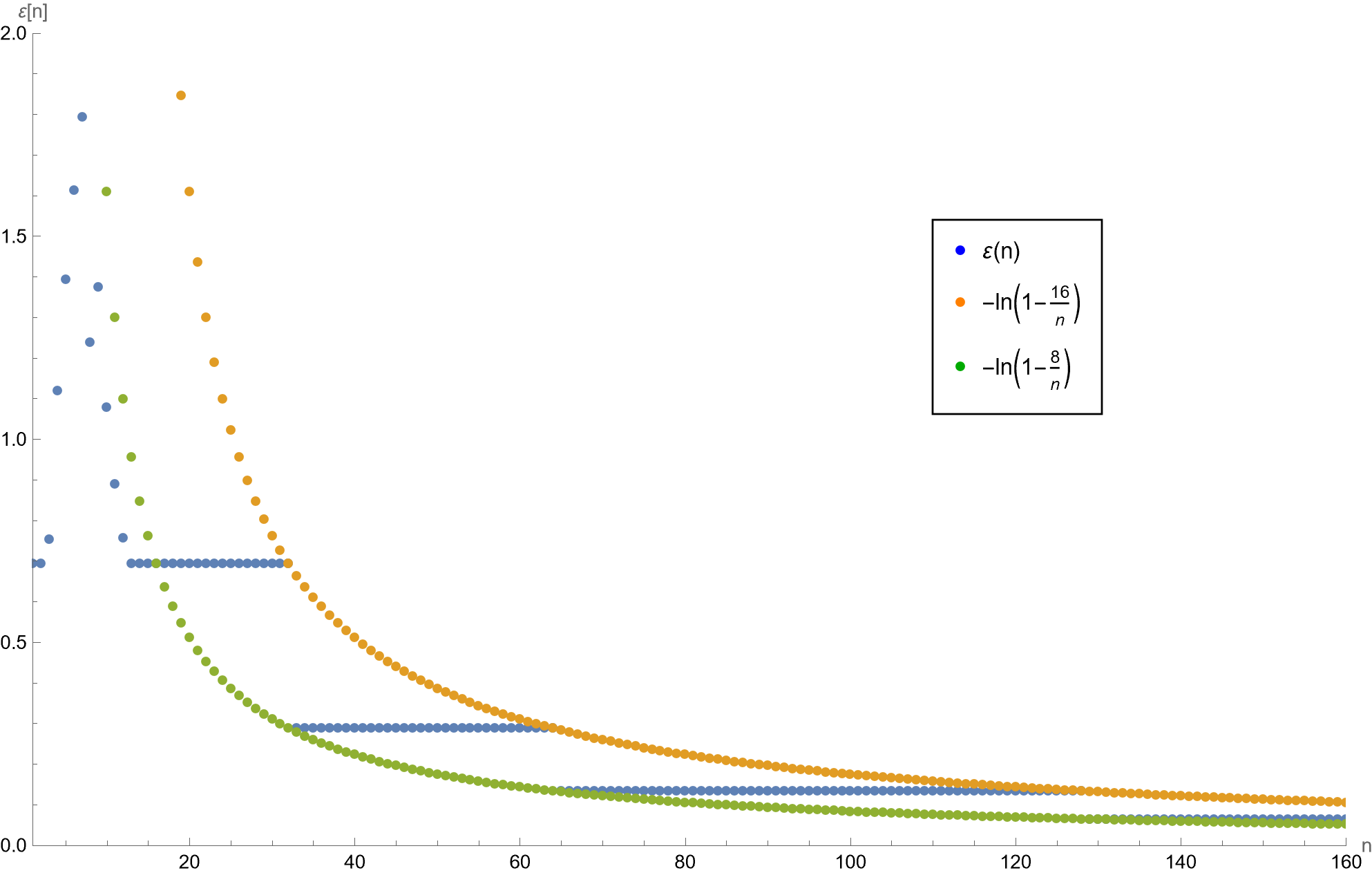}
    \caption{Exact values of $\varepsilon(n)$ parameter for $n\ls 160$ compared with plots of sequences $-\ln(1-{16}/{n})$ and $-\ln(1-{8}/{n})$.}
    {\label{fig:epsilon}}
\end{figure}

In the Figure \ref{fig:epsilon} one can see that values of $\varepsilon(n)$ are strictly between sequences $-\ln(1-~{8}/{n})$ and $-\ln(1-~{16}/{n})$ for $n\in[17:160]$. We can also observe that $\varepsilon(n)~\approx~2^{4-\lceil\log(n)\rceil}$ in this interval. Note that $\lceil\log(n)\rceil\ls 4$ for $n\ls 16$, so $\lceil\log(n)\rceil -4<1$, but $M$ is always positive. This can justify the chaotic behaviour of the process for $n\ls 16$.
Nevertheless, Figure \ref{fig:epsilon} affirms the quality of $\varepsilon(n)$ parameter established in Theorem \ref{thm:main}. 

Moreover, let us mention that for the strategy used in the proof, we cannot pick smaller $\varepsilon$ of the same form:
\begin{observation}\label{lowerObs}
Assume that for $I_n$, defined in (\ref{eq:In}), $\delta=\Pr{M_n\notin I_n}$ and $\varepsilon(n)$, defined via (\ref{eq:epsilon}), is of the form $|-\ln(1- \frac{c}{n})|$. Then $c=16$ is optimal constant in Theorem~\ref{thm:main} (it cannot be improved). Indeed, see that the bound is reached for $n=32$:
$$
\dfrac{p_{33,1}}{p_{32,1}}=\frac{1}{2}=1-\frac{16}{32}~.
$$
\end{observation}
As can be seen in Figure \ref{fig:epsilon}, the minimal $c$ for $\varepsilon(n)$, where $n$ is of the form $2^k$, for $k>5$, is also very close to $16$.

\subsection{General  result on Morris' Counter privacy (Theorem~\ref{thm:addition})}
\label{proof2}

Remark that $\lceil c\log(\ln(n))\rceil\gs 1$ can always be guaranteed, when $n$ is large enough.

\begin{proof}
For our convenience, let us denote $\rho:=\lceil c\log(\ln(n))\rceil$. We assume that $\rho\gs 1$.
First, we show that $\delta_1^{\ast}:=\PR{M_n\ls \lceil\log(n)\rceil - \rho -1}=O\left(n^{-\left(\ln(n)\right)^{c-1}}\right)$. The proof is analogous to the one of Lemma \ref{lem:lowerdelta} (we omit similar parts).
Indeed,
\begin{multline*}
\delta_1^{\ast}
=\sum_{l=1}^{\left\lceil\log(n)\right\rceil-\rho-1}\PR{M_n=l}~
\ls R^2 (2\sqrt{2}+2) \sum_{k=\rho}^{\infty}\exp(-2^k)\\
\ls R^2 (2\sqrt{2}+2) \sum_{k=1}^{\infty}\exp(-2^{\rho} k)\\
\ls R^2 (2\sqrt{2}+2)\frac{\exp(-(\ln(n))^c)}{1-\exp(-(\ln(n))^c)}=O\left(n^{-\left(\ln(n)\right)^{c-1}}\right)~.
\end{multline*}
Now, we are going to prove $\delta_2^{\ast}:=\PR{M_n\geqslant \lceil\lg(n)\rceil + \rho +1}=O\left(n^{-1}\left(\ln(n)\right)^{-c}\right)$.
We use a similar notation and technique as in the proof of Lemma \ref{lem:upperdelta}, but this time we utilise a discrete discrete two-dimensional simplex $(\rho -1)$: 
$$
S_{k}^{(\rho-1)}=\left\{\bar{l}=(l_1,l_2,\ldots,l_{\rho-1})\in\NN_0^{\rho-1}: \sum_{i=1}^{\rho-1} l_i\ls k\right\}~.
$$
We couple $(M_n)_n$ with the same process $(X_n)_n$ as in Lemma \ref{lem:upperdelta}. Roughly speaking, $X_0=1$ and $X_n$ almost always increments by $1$ until $n=\lceil\log(n)\rceil +1$ and it further follows the same incrementation rule as the Morris counter. Then
\begin{align*}
\delta_2^{\ast}&\ls \PR{X_n\gs \left\lceil \log(n)\right\rceil + \rho +1}
=\sum_{\bar{l}\in S_{m-\rho+1}^{(\rho-1)}} \prod_{i=1}^{\rho-1} \nu_{i+1}^{l_i} \mu_i \ls \sum_{\bar{l}\in S_{m-\rho+1}^{(\rho-1)}} \prod_{i=1}^{\rho-1} \mu_i \\
&= \sum_{\bar{l}\in S_{m-\rho+1}^{(\rho-1)}} 2^{-[(\rho-1)\left\lceil\log(n)\right\rceil +\sum_{i=2}^{\rho} i]}\\
&=\sum_{k=0}^{m-\rho+1} \binom{k+\rho-1}{\rho-2} 2^{-[(\rho-1)\left\lceil\log(n)\right\rceil +\frac{(\rho+2)(\rho-1)}{2}]}\\
&\ls m^{\rho-2} n^{1-\rho} 2^{-\frac{\rho^2+\rho-2}{2}} \ls n^{-1} 2^{-\rho+1}=O\left(n^{-1}\left(\ln(n)\right)^{-c}\right)~.
\end{align*}
Therefore $\PR{M_n\notin J_n(c)}=\delta_1^{\ast}+\delta_2^{\ast}=O\left(n^{-\left(\ln(n)\right)^{c-1}}+ n^{-1}\left(\ln(n)\right)^{-c}\right)$

In addition, we would like to consider fractions $\frac{p_{n+1,k}}{p_{n,k}}$ for $k\in J_n(c)$ as in the proof of Theorem \ref{thm:morriseps}.
Indeed
$$
1-2^{-k}\ls \frac{p_{n+1,k}}{p_{n,k}}=1-2^{-k}+2^{-k+1}\frac{p_{n,k-1}}{p_{n,k}}~.
$$
We are going to use another formula from \cite{flajolet1985approximate}. For any $n\in\NN$ and $k\in [1:n+1]$,
$$
p_{n,k}=2^{-\frac{k(k-1)}{2}}\sum_{\bar{l}\in S_{n-k+1}^{(k)}} \prod_{i=1}^{k} \left(1-2^{-i}\right)^{l_i}~.
$$
Let us denote the above sum by $\varsigma_k(n-k+1)$.
We note that 
$$
2^{-k+1}\frac{p_{n,k-1}}{p_{n,k}}=2^{-k+1}\frac{2^{-\frac{(k-2)(k-1)}{2}}\varsigma_{k-1}(n-k+2)}{2^{-\frac{k(k-1)}{2}}\varsigma_k(n-k+1)}=\frac{\varsigma_{k-1}(n-k+2)}{\varsigma_k(n-k+1)}~.
$$
Realise that $\varsigma_{k-1}(n-k+2)\ls \varsigma_{k-1}(n-k+1) \sum_{i=1}^{k-1}\left(1-2^{-i}\right)=\varsigma_{k-1}(n-k+1) \left(k-2+2^{k-1}\right)$. This follows from the fact that  each summand of $\varsigma_{k-1}(n-k+2)$ can be obtained from some summands of $\varsigma_{k-1}(n-k+1)$ by multiplication by one of the terms $\left(1-2^{-i}\right)$.
Moreover, note that $\varsigma_k(n-k+1)$ has $\binom{n-k+1 +(k-1)}{k-1}=\binom{n}{k-1}$ summands and, similarly, $\varsigma_{k-1}(n-k+1)$ has 
$\binom{n-1}{k-2}$ summands.
One can briefly see that a function $f(i)=\left(1-2^{-i}\right)$ is increasing, hence
\begin{align*}
\varsigma_{k}(n-k+1)&\gs \sum_{\bar{l}\in S_{n-k+1}^{(k)}} \left(1-2^{-k+1}\right)^{l_{k-1}+l_k}\prod_{i=1}^{k-2} \left(1-2^{-i}\right)^{l_i}\\
&= \sum_{\bar{l}\in S_{n-k+1}^{(k-1)}} (l_{k-1} +1)\prod_{i=1}^{k-1} \left(1-2^{-i}\right)^{l_i}~.
\end{align*}
Due to the monotonicity of $f$, one can use cascading substitutions: some of $f(k)$ by $f(k-1)$, then some of $f(k-1)$ by $f(k-2)$ etc., in order to balance the numbers of all the occurring summands, what provides:
$$
\varsigma_{k}(n-k+1)\gs \sum_{\bar{l}\in S_{n-k+1}^{(k-1)}} \frac{\binom{n}{k-1}}{\binom{n-1}{k-2}}\prod_{i=1}^{k-1} \left(1-2^{-i}\right)^{l_i}=\frac{n}{k-1} \varsigma_{k-1}(n-k+1)~.
$$
Therefore
$\varsigma_{k-1}(n-k+2)\ls \frac{\left(k-2+2^{k-1}\right)(k-1)}{n}\varsigma_{k}(n-k+1)$ and finally we obtain
$$
\frac{p_{n+1,k}}{p_{n,k}}\ls 1-2^{-k} +\frac{\left(k-2+2^{k-1}\right)(k-1)}{n}<1+\frac{(k-1)^2}{n}~.
$$
When $k\in J_n(c)$, then 
\begin{multline*}
\exp\left(-O\left(\frac{\log(n)}{n}\right)\right)=1-2^{-\lceil\log(n)\rceil +\lceil\log\ln(n)\rceil}\\
<\frac{p_{n+1,k}}{p_{n,k}}<\\
1+\frac{(\lceil\log(n)\rceil +\log\ln(n))^2}{n}=\exp\left(O\left(\frac{(\log(n))^2}{n}\right)\right)~.
\end{multline*}
This shows that $\varepsilon(n)=O\left(\frac{(\log(n))^2}{n}\right)$.
\end{proof}

\subsection{Main result for MaxGeo Counter (Theorem~\ref{maxGeoTheorem})}
\label{proof3}

\begin{proof}
We have $n$ increment requests, which influence the value of MaxGeo Counter $M$. Then the result of the mechanism is $X = \max(X_1,\ldots,X_n)$, where $X_i \sim \mathrm{Geo}({1}/{2})$ are pairwise independent. First, we observe that if $n=m$, the counter trivially satisfies differential privacy, as the probability distribution of $X$ does not change. From now on, we assume that $|n-m| = 1$. See that
$$
\PR{X \ls l} = \prod_{i=1}^n \PR{X_i \ls l} = \left(\PR{X_1 \ls l}\right)^n = \left(1-\frac{1}{2^l}\right)^n = \left(\frac{2^l - 1}{2^l}\right)^n~.
$$
Furthermore
\begin{align*}
\PR{\max(X_1,\ldots,X_n) = l} = \PR{X = l} = \PR{X \ls l} - \PR{X \ls (l-1)}\\
=\left(\frac{2^l - 1}{2^l}\right)^n - \left(\frac{2^{l-1} - 1}{2^{l-1}}\right)^n = \frac{\left(2^l - 1\right)^n - \left(2^l - 2\right)^n}{2^{l\cdot n}}~.
\end{align*}
Now we need to calculate the following expression
\begin{align*}
\frac{\PR{\max(X_1,...X_n) = l}}{\PR{\max(X_1,...X_n,X_{n+1}) = l}} &= \dfrac{\dfrac{\left(2^l - 1\right)^n - \left(2^l - 2\right)^n}{2^{l\cdot n}}}{\dfrac{\left(2^l - 1\right)^{n+1} - \left(2^l - 2\right)^{n+1}}{2^{l\cdot (n+1)}}} \\
&=\frac{\left(\left(2^l - 1\right)^n - (2^l - 2)^n\right) \cdot 2^l}{\left(2^l - 1\right)^{n+1} - (2^l - 2)^{n+1}} \\
&= \frac{2^l}{2^l - 1} \cdot \frac{\left(\left(2^l - 1\right)^n - (2^l - 2)^n\right)}{\left(\left(2^l - 1\right)^{n} - \frac{(2^l - 2)^{n+1}}{2^l-1}\right)} \\
&\ls \frac{2^l}{2^l - 1} \cdot \frac{\left(\left(2^l - 1\right)^n - (2^l - 2)^n\right)}{\left(\left(2^l - 1\right)^{n} - \frac{(2^l - 2)^{n+1}}{2^l-2}\right)} \\
&= \frac{2^l}{2^l - 1} = 1 + \frac{1}{2^l-1}~.
\end{align*}
For fixed $\varepsilon$ we need to satisfy the following inequality
$$
\left| \ln\left(\frac{P(\max(X_1,...X_n) = l)}{P(\max(X_1,...X_n,X_{n+1}) = l)}\right) \right| \ls \varepsilon~,
$$
which gives
\begin{equation}\label{eq:logepsilon}
\ln\left(1 + \frac{1}{2^l - 1}\right) 
 \ls \varepsilon~.
\end{equation}
We can see from (\ref{eq:logepsilon}) that the greater $l$ is, the smaller $\varepsilon$ can be. Moreover, inequality (\ref{eq:logepsilon}) is true for $l \gs l_{\varepsilon}$. Therefore, we must ensure $P(X \ls l_{\varepsilon}) \ls \delta$. See that 
$$
\PR{X \leqslant l_{\varepsilon}} = \left( 1 - 2^{-l_{\varepsilon}} \right)^n ~.
$$
It is easy to see that the above decreases as $n$ increases. Then
$$
\left( 1 - 2^{-l_{\varepsilon}} \right)^n \ls \delta \iff n \gs \frac{\ln(\delta)}{\ln(1-2^{-l_{\varepsilon}})}\approx -\frac{\ln(\delta)}{\varepsilon},
$$
where the approximation is the result of the substitution of $l_{\varepsilon}$ without ceiling.
\end{proof}

\section{Conclusions and Future Work}\label{sect:conclusion}

In this paper, we have investigated probabilistic counters from a privacy-protection perspective. We have shown that Morris Counter and MaxGeo Counter inherently guarantee differential privacy from the mechanism itself, provided that there is at least a small fixed number of increment requests. Otherwise, the counter has too low a value and, intuitively, the result is not randomised enough. We have also shown that the constant in our Morris Counter result cannot be improved further.

We have shown how to perform data aggregation, namely a distributed survey, in a privacy-preserving manner using probabilistic counters.
We clarified that this type of solution is especially efficient when one cares about memory resources, like in many Big Data related problems.
Note that the security model in this paper was somewhat optimistic. Unfortunately, in such a setting, there is little incentive to use them other than when we already have them deployed and working as aggregators due to e.g., memory-efficiency requirements. However, this would change tremendously if we weakened these assumptions. This seems a very promising way to continue our research from this paper. Namely, we focused on privacy and can still not weaken the security assumptions and allow the Adversary to extract information from channels between users and the aggregator. That would put us in the so-called Local Model, where each user is responsible for the data randomisation. However, such an approach requires us to be able to perform probabilistic counter in an oblivious manner, which, to the best of our knowledge, was not explored before. 

In Subsection \ref{ssec:morris}, we have mentioned the generalisation of the Morris Counter (for bases $a>1$). Analysis of privacy properties of such variants of Morris Counters and various probabilistic counters presented, for example, in \cite{2009arXiv0904.3062C}, \cite{fuchs:hal-01197238} may also be promising directions of further research. 


In this paper, we focus on the standard definition of differential privacy. However, there is also an issue of preservation of differential privacy for requests given by a group of $k$ users   or one individual sending up to $k$ dependent  requests over time. This  can be described in the language of the so-called $(\varepsilon,\delta)-k$-DP (see the "group privacy scenario"  in \cite{DworkAlgo}) . A group of people may tend to behave in the same manner, so they may send $k$ requests in a row. Especially this ''group'' may be represented by a single person colluding with the Adversary. It is worth mentioning that this type of generalisation creates an opportunity to modify probabilistic counters so that each incrementation request executes the update request multiple times to reduce the variance of the rescaled estimator. Intuitively, this extension should be especially efficient in preserving the standard differential privacy property when $\varepsilon(n)=\frac{c}{n} + o(n^{-1})$ (as a parameter of standard differential privacy), because both $c$ and $n$ should scale with $k$ linearly. Hence, the next challenging problem is to show that Morris and MaxGeo Counters satisfy the $k$-DP property with similar privacy parameters.
\corr{It can also be one individual sending many requests over time, which I think is more applicable to the real world applications. Also this is generally called "group privacy"}{
}{}

The Morris Counter and the MaxGeo Counter are considered the most popular probabilistic counters. However, the results of this paper shed new light on the properties of the probabilistic counter, in general. There is a possibility to provide analogous differential privacy properties for other probabilistic counters. Moreover, this paper enables the provision of differentially private algorithms for other applications, especially those based on Morris or MaxGeo Counter.
For example, in Section \ref{sect:counters} we mentioned PCSA and HyperLogLog Counters together with their variances, which can be manually adjusted to the applications. The proper choice of the $m$ parameter implies an exchange of memory usage to improve the accuracy of the estimation. We have mentioned that these counters' differential privacy parameters can be obtained via Observation \ref{obs:maxDP}. However, such a direct result may not be satisfying. Hence a more precise calculation is needed. For example, Observation \ref{obs:maxDP} may be used again with some concentration inequalities.

\acknowledgements
\label{sec:ack}
The authors sincerely thank the reviewers and editors for their significant time and effort invested in editing and improving this work.
This work was supported by the Polish National Science Center grant number UMO-2018/29/B/ST6/02969.

\bibliographystyle{abbrvnat}
\bibliography{new-dmtcs}

@misc{mcmahan2018learningdifferentiallyprivaterecurrent,
      title={Learning Differentially Private Recurrent Language Models}, 
      author={H. Brendan McMahan and Daniel Ramage and Kunal Talwar and Li Zhang},
      year={2018},
      eprint={1710.06963},
      archivePrefix={arXiv},
      primaryClass={cs.LG},
      url={https://arxiv.org/abs/1710.06963}, 
}

@inproceedings{Wang_2019, series={CIKM ’19},
   title={Privacy-preserving Crowd-guided AI Decision-making in Ethical Dilemmas},
   url={http://dx.doi.org/10.1145/3357384.3357954},
   DOI={10.1145/3357384.3357954},
   booktitle={Proceedings of the 28th ACM International Conference on Information and Knowledge Management},
   publisher={ACM},
   author={Wang, Teng and Zhao, Jun and Yu, Han and Liu, Jinyan and Yang, Xinyu and Ren, Xuebin and Shi, Shuyu},
   year={2019},
   month=nov, pages={1311–1320},
   collection={CIKM ’19} }

@misc{papernot2017semisupervisedknowledgetransferdeep,
      title={Semi-supervised Knowledge Transfer for Deep Learning from Private Training Data}, 
      author={Nicolas Papernot and Martín Abadi and Úlfar Erlingsson and Ian Goodfellow and Kunal Talwar},
      year={2017},
      eprint={1610.05755},
      archivePrefix={arXiv},
      primaryClass={stat.ML},
      url={https://arxiv.org/abs/1610.05755}, 
}

@misc{nelson2022optimal,
      title={Optimal bounds for approximate counting}, 
      author={Jelani Nelson and Huacheng Yu},
      year={2022},
      eprint={2010.02116},
      archivePrefix={arXiv},
      primaryClass={cs.DS}
}

@inproceedings{cichon2011approximate,
  title={Approximate counters for flash memory},
  author={Cichon, Jacek and Macyna, Wojciech},
  booktitle={2011 IEEE 17th International Conference on Embedded and Real-Time Computing Systems and Applications},
  volume={1},
  pages={185--189},
  year={2011},
  organization={IEEE}
}

@article{desfontaines2019cardinality,
  title={Cardinality estimators do not preserve privacy},
  author={Desfontaines, Damien and Lochbihler, Andreas and Basin, David},
  journal={Proceedings on Privacy Enhancing Technologies},
  volume={2019},
  number={2},
  pages={26--46},
  year={2019},
  publisher={Sciendo}
}

@inproceedings{van2009probabilistic,
  title={Probabilistic counting with randomized storage},
  author={Van Durme, Benjamin and Lall, Ashwin},
  booktitle={Twenty-First International Joint Conference on Artificial Intelligence},
  year={2009}
}

@inproceedings{dice2013scalable,
  title={Scalable statistics counters},
  author={Dice, Dave and Lev, Yossi and Moir, Mark},
  booktitle={Proceedings of the twenty-fifth annual ACM symposium on Parallelism in algorithms and architectures},
  pages={43--52},
  year={2013}
}

@inproceedings{flajolet2007hyperloglog,
  title={Hyperloglog: the analysis of a near-optimal cardinality estimation algorithm},
  author={Flajolet, Philippe and Fusy, {\'E}ric and Gandouet, Olivier and Meunier, Fr{\'e}d{\'e}ric},
  booktitle={AofA: Analysis of Algorithms},
  pages={137--156},
  year={2007},
  organization={Discrete Mathematics and Theoretical Computer Science}
}

@article{morris1978counting,
  title={Counting large numbers of events in small registers},
  author={Morris, Robert},
  journal={Communications of the ACM},
  volume={21},
  number={10},
  pages={840--842},
  year={1978},
  publisher={ACM}
}

@article{flajolet1985approximate,
  title={Approximate counting: a detailed analysis},
  author={Flajolet, Philippe},
  journal={BIT Numerical Mathematics},
  volume={25},
  number={1},
  pages={113--134},
  year={1985},
  publisher={Springer}
}

@article{riley2006probabilistic,
  title={Probabilistic counter updates for predictor hysteresis and bias},
  author={Riley, Nicholas and Zilles, Craig},
  journal={IEEE Computer Architecture Letters},
  volume={5},
  number={1},
  pages={18--21},
  year={2006},
  publisher={IEEE}
}

@article{flajolet1985probabilistic,
  title={Probabilistic counting algorithms for data base applications},
  author={Flajolet, Philippe and Martin, G Nigel},
  journal={Journal of computer and system sciences},
  volume={31},
  number={2},
  pages={182--209},
  year={1985},
  publisher={Elsevier}
}

@article{eisenberg2008expectation,
  title={On the expectation of the maximum of IID geometric random variables},
  author={Eisenberg, Bennett},
  journal={Statistics \& Probability Letters},
  volume={78},
  number={2},
  pages={135--143},
  year={2008},
  publisher={Elsevier}
}

@article{szpankowski1990yet,
  title={Yet another application of a binomial recurrence order statistics},
  author={Szpankowski, Wojciech and Rego, Vernon},
  journal={Computing},
  volume={43},
  number={4},
  pages={401--410},
  year={1990},
  publisher={Springer}
}

@inproceedings{durand2003loglog,
  title={Loglog counting of large cardinalities},
  author={Durand, Marianne and Flajolet, Philippe},
  booktitle={European Symposium on Algorithms},
  pages={605--617},
  year={2003},
  organization={Springer}
}

@incollection{bollobas1998random,
  title={Random graphs},
  author={Bollob{\'a}s, B{\'e}la},
  booktitle={Modern Graph Theory},
  pages={215--252},
  year={1998},
  publisher={Springer}
}

@article{DworkAlgo,
  title={The algorithmic foundations of differential privacy},
  author={Dwork, Cynthia and Roth, Aaron},
  journal={Foundations and Trends in Theoretical Computer Science},
  volume={9},
  number={3-4},
  pages={211--407},
  year={2014}
}

@inproceedings{dwork2006our,
  title={Our data, ourselves: Privacy via distributed noise generation},
  author={Dwork, Cynthia and Kenthapadi, Krishnaram and McSherry, Frank and Mironov, Ilya and Naor, Moni},
  booktitle={Annual International Conference on the Theory and Applications of Cryptographic Techniques},
  pages={486--503},
  year={2006},
  organization={Springer}
}

@inproceedings{dwork2009differential,
  title={Differential privacy and robust statistics.},
  author={Dwork, Cynthia and Lei, Jing},
  booktitle={STOC},
  volume={9},
  pages={371--380},
  year={2009}
}

@inproceedings{dwork2010differential,
  title={Differential privacy under continual observation},
  author={Dwork, Cynthia and Naor, Moni and Pitassi, Toniann and Rothblum, Guy N},
  booktitle={Proceedings of the forty-second ACM symposium on Theory of computing},
  pages={715--724},
  year={2010},
  organization={ACM}
}

@inproceedings{narayanan2009anonymizing,
  title={De-anonymizing social networks},
  author={Narayanan, Arvind and Shmatikov, Vitaly},
  booktitle={Security and Privacy, 2009 30th IEEE Symposium on},
  pages={173--187},
  year={2009},
  organization={IEEE}
}

@article{narayanan2010myths,
  title={Myths and fallacies of personally identifiable information},
  author={Narayanan, Arvind and Shmatikov, Vitaly},
  journal={Communications of the ACM},
  volume={53},
  number={6},
  pages={24--26},
  year={2010},
  publisher={ACM}
}

@inproceedings{Dwork06,
  author    = {Cynthia Dwork},
  title     = {Differential Privacy},
  booktitle = {Automata, Languages and Programming, 33rd International Colloquium,
               {ICALP} 2006},
  pages     = {1--12},
  year      = {2006},
  bibsource = {dblp computer science bibliography, http://dblp.org}
}

@inproceedings{dwork2006calibrating,
  title={Calibrating noise to sensitivity in private data analysis},
  author={Dwork, Cynthia and McSherry, Frank and Nissim, Kobbi and Smith, Adam},
  booktitle={TCC},
  volume={3876},
  pages={265--284},
  year={2006},
  organization={Springer}
}

@article{AniaPio,
  author    = {Marek Klonowski and
               Ania M. Piotrowska},
  title     = {Light-weight and secure aggregation protocols based on Bloom filters},
  journal   = {Comput. Secur.},
  volume    = {72},
  pages     = {107--121},
  year      = {2018},
  url       = {https://doi.org/10.1016/j.cose.2017.08.015},
  doi       = {10.1016/j.cose.2017.08.015},
  bibsource = {dblp computer science bibliography, https://dblp.org}
}

@article{Crippa:1997:QMP:2781893.2781980,
 author = {Crippa, Davide and Simon, Klaus},
 title = {Q-distributions and Markov Processes},
 journal = {Discrete Math.},
 issue_date = {June 1997},
 volume = {170},
 number = {1},
 month = jun,
 year = {1997},
 issn = {0012-365X},
 pages = {81--98},
 numpages = {18},
 acmid = {2781980},
 publisher = {Elsevier Science Publishers B. V.},
}

@book{frieze_karonski_2015, 
place={Cambridge}, 
title={Introduction to Random Graphs},
publisher={Cambridge University Press}, 
author={Frieze, Alan and Karonski, Michał}, 
year={2015}
}

@inproceedings{heule2013hyperloglog,
  title={HyperLogLog in practice: algorithmic engineering of a state of the art cardinality estimation algorithm},
  author={Heule, Stefan and Nunkesser, Marc and Hall, Alexander},
  booktitle={Proceedings of the 16th International Conference on Extending Database Technology},
  pages={683--692},
  year={2013}
}

@inproceedings{2009arXiv0904.3062C,
  title={Approximate counting with a floating-point counter},
  author={Cs{\H{u}}r{\"o}s, Mikl{\'o}s},
  booktitle={International Computing and Combinatorics Conference},
  pages={358--367},
  year={2010},
  organization={Springer}
}

@Article{Gronemeier2009,
author="Gronemeier, Andr{\'e}
and Sauerhoff, Martin",
title="Applying Approximate Counting for Computing the Frequency Moments of Long Data Streams",
journal="Theory of Computing Systems",
year="2009",
volume="44",
number="3",
pages="332--348",
issn="1433-0490"
}

@inproceedings{fuchs:hal-01197238,
  TITLE = {{Approximate Counting via the Poisson-Laplace-Mellin Method}},
  AUTHOR = {Fuchs, Michael and Lee, Chung-Kuei and Prodinger, Helmut},
  BOOKTITLE = {{23rd International Meeting on Probabilistic, Combinatorial, and Asymptotic Methods in the Analysis of Algorithms (AofA'12)}},
  PUBLISHER = {{Discrete Mathematics and Theoretical Computer Science}},
  PAGES = {13-28},
  YEAR = {2012},
}

@article{NeymanEstimation37,
 ISSN = {00804614},
 author = {J. Neyman},
 journal = {Philosophical Transactions of the Royal Society of London. Series A, Mathematical and Physical Sciences},
 number = {767},
 pages = {333--380},
 publisher = {The Royal Society},
 title = {Outline of a Theory of Statistical Estimation Based on the Classical Theory of Probability},
 volume = {236},
 year = {1937}
}

@misc{Mathematica,
  author = {Wolfram Research{,} Inc.},
  title = {Mathematica, {V}ersion 11.3}
}

@inproceedings{StreamHLL,
author = {Ting, Daniel},
title = {Streamed Approximate Counting of Distinct Elements: Beating Optimal Batch Methods},
year = {2014},
isbn = {9781450329569},
publisher = {Association for Computing Machinery},
address = {New York, NY, USA}, 
doi = {10.1145/2623330.2623669},
booktitle = {Proceedings of the 20th ACM SIGKDD International Conference on Knowledge Discovery and Data Mining},
pages = {442–451},
numpages = {10},
location = {New York, New York, USA},
series = {KDD '14}
}

@ARTICLE{StreamHIP,  author={E. {Cohen}},  journal={IEEE Transactions on Knowledge and Data Engineering},   title={All-Distances Sketches, Revisited: HIP Estimators for Massive Graphs Analysis},   year={2015},  volume={27},  number={9},  pages={2320-2334},  doi={10.1109/TKDE.2015.2411606}}

@ARTICLE{FastRAQ,  author={X. {Yun} and G. {Wu} and G. {Zhang} and K. {Li} and S. {Wang}},  journal={IEEE Transactions on Cloud Computing},   title={FastRAQ: A Fast Approach to Range-Aggregate Queries in Big Data Environments},   year={2015},  volume={3},  number={2},  pages={206-218},  doi={10.1109/TCC.2014.2338325}}

@inproceedings{TingBillionsDatasets,
author = {Ting, Daniel},
title = {Approximate Distinct Counts for Billions of Datasets},
year = {2019},
isbn = {9781450356435},
publisher = {Association for Computing Machinery},
address = {New York, NY, USA},
doi = {10.1145/3299869.3319897},
booktitle = {Proceedings of the 2019 International Conference on Management of Data},
pages = {69–86},
numpages = {18},
location = {Amsterdam, Netherlands},
series = {SIGMOD '19}
}

@ARTICLE {ICEBuckets,
author = {G. Einziger and B. Fellman and R. Friedman and Y. Kassner},
journal = {IEEE/ACM Transactions on Networking},
title = {ICE Buckets: Improved Counter Estimation for Network Measurement},
year = {2018},
volume = {26},
number = {03},
issn = {1558-2566},
pages = {1165-1178},
doi = {10.1109/TNET.2018.2822734},
publisher = {IEEE Computer Society},
address = {Los Alamitos, CA, USA},
month = {may}
}

@inproceedings{Indyk2003Woodruff,
author = {Indyk, P. and Woodruff, D.},
year = {2003},
month = {11},
pages = {283- 288},
title = {Tight lower bounds for the distinct elements problem},
isbn = {0-7695-2040-5},
journal = {Annual Symposium on Foundations of Computer Science - Proceedings},
doi = {10.1109/SFCS.2003.1238202}
}

@inproceedings{ANF2002,
author = {Palmer, Christopher R. and Gibbons, Phillip B. and Faloutsos, Christos},
title = {ANF: A Fast and Scalable Tool for Data Mining in Massive Graphs},
year = {2002},
isbn = {158113567X},
publisher = {Association for Computing Machinery},
address = {New York, NY, USA},
url = {https://doi.org/10.1145/775047.775059},
doi = {10.1145/775047.775059},
booktitle = {Proceedings of the Eighth ACM SIGKDD International Conference on Knowledge Discovery and Data Mining},
pages = {81–90},
numpages = {10},
location = {Edmonton, Alberta, Canada},
series = {KDD '02}
}

@article{DBLP:journals/jcisd/SwamidassB07a,
  author    = {S. Joshua Swamidass and
               Pierre Baldi},
  title     = {Mathematical Correction for Fingerprint Similarity Measures to Improve
               Chemical Retrieval},
  journal   = {J. Chem. Inf. Model.},
  volume    = {47},
  number    = {3},
  pages     = {952--964},
  year      = {2007},
  url       = {https://doi.org/10.1021/ci600526a},
  doi       = {10.1021/ci600526a},
  bibsource = {dblp computer science bibliography, https://dblp.org}
}

@article{Bloom1970,
author = {Bloom, Burton H.},
title = {Space/Time Trade-Offs in Hash Coding with Allowable Errors},
year = {1970},
issue_date = {July 1970},
publisher = {Association for Computing Machinery},
address = {New York, NY, USA},
volume = {13},
number = {7},
issn = {0001-0782},
url = {https://doi.org/10.1145/362686.362692},
doi = {10.1145/362686.362692},
journal = {Commun. ACM},
month = jul,
pages = {422–426},
numpages = {5}
}

@article{ALON1999137,
title = "The Space Complexity of Approximating the Frequency Moments",
journal = "Journal of Computer and System Sciences",
volume = "58",
number = "1",
pages = "137 - 147",
year = "1999",
issn = "0022-0000",
doi = "https://doi.org/10.1006/jcss.1997.1545",
url = "http://www.sciencedirect.com/science/article/pii/S0022000097915452",
author = "Noga Alon and Yossi Matias and Mario Szegedy"
}

@inproceedings{DBLP:conf/focs/Indyk00,
  author    = {Piotr Indyk},
  title     = {Stable Distributions, Pseudorandom Generators, Embeddings and Data
               Stream Computation},
  booktitle = {41st Annual Symposium on Foundations of Computer Science, {FOCS} 2000,
               12-14 November 2000, Redondo Beach, California, {USA}},
  pages     = {189--197},
  publisher = {{IEEE} Computer Society},
  year      = {2000},
  url       = {https://doi.org/10.1109/SFCS.2000.892082},
  doi       = {10.1109/SFCS.2000.892082},
  bibsource = {dblp computer science bibliography, https://dblp.org}
}

@misc{HyperBitBit,
author = {Robert Sedgewick},
title = {Cardinality Estimation},
year = {2018},
url = {https://www.cs.princeton.edu/~rs/talks/CardinalityX.pdf}
}

@inproceedings{Baquero2009,
author = {Baquero, Carlos and Almeida, Paulo and Menezes, Raquel},
year = {2009},
month = {01},
pages = {88-93},
title = {Fast Estimation of Aggregates in Unstructured Networks},
doi = {10.1109/ICAS.2009.31}
}

@article{JCIGotfryd,
author = {Cicho\'{n}, Jacek and Gotfryd, Karol},
title = {Average Counting via Approximate Histograms},
year = {2018},
issue_date = {July 2018},
publisher = {Association for Computing Machinery},
address = {New York, NY, USA},
volume = {14},
number = {2},
issn = {1550-4859},
url = {https://doi.org/10.1145/3177922},
doi = {10.1145/3177922},
journal = {ACM Trans. Sen. Netw.},
month = mar,
articleno = {8},
numpages = {32}
}

@article{10.1145/198429.198435,
author = {Li, Kim-Hung},
title = {Reservoir-Sampling Algorithms of Time Complexity $O(n(1 + Log(N/n)))$},
year = {1994},
issue_date = {Dec. 1994},
publisher = {Association for Computing Machinery},
address = {New York, NY, USA},
volume = {20},
number = {4},
issn = {0098-3500},
url = {https://doi.org/10.1145/198429.198435},
doi = {10.1145/198429.198435},
journal = {ACM Trans. Math. Softw.},
month = dec,
pages = {481–493},
numpages = {13}
}

@article{PanChoi,
  author    = {Seung Geol Choi and
               Dana Dachman{-}Soled and
               Mukul Kulkarni and
               Arkady Yerukhimovich},
  title     = {Differentially-Private Multi-Party Sketching for Large-Scale Statistics},
  journal   = {Proc. Priv. Enhancing Technol.},
  volume    = {2020},
  number    = {3},
  pages     = {153--174},
  year      = {2020},
  url       = {https://doi.org/10.2478/popets-2020-0047},
  doi       = {10.2478/popets-2020-0047},
  bibsource = {dblp computer science bibliography, https://dblp.org}
}

@inproceedings{ADAMS,
  author    = {Adam D. Smith and
               Shuang Song and
               Abhradeep Thakurta},
  editor    = {Hugo Larochelle and
               Marc'Aurelio Ranzato and
               Raia Hadsell and
               Maria{-}Florina Balcan and
               Hsuan{-}Tien Lin},
  title     = {The Flajolet-Martin Sketch Itself Preserves Differential Privacy:
               Private Counting with Minimal Space},
  booktitle = {Advances in Neural Information Processing Systems 33: Annual Conference
               on Neural Information Processing Systems 2020, NeurIPS 2020, December
               6-12, 2020, virtual},
  year      = {2020},
  url       = {https://proceedings.neurips.cc/paper/2020/hash/e3019767b1b23f82883c9850356b71d6-Abstract.html},
  bibsource = {dblp computer science bibliography, https://dblp.org}
}

@inproceedings{Kamil2,
  author    = {Nina Mishra and
               Mark Sandler},
  editor    = {Stijn Vansummeren},
  title     = {Privacy via pseudorandom sketches},
  booktitle = {Proceedings of the Twenty-Fifth {ACM} {SIGACT-SIGMOD-SIGART} Symposium
               on Principles of Database Systems, June 26-28, 2006, Chicago, Illinois,
               {USA}},
  pages     = {143--152},
  publisher = {{ACM}},
  year      = {2006},
  url       = {https://doi.org/10.1145/1142351.1142373},
  doi       = {10.1145/1142351.1142373},
  bibsource = {dblp computer science bibliography, https://dblp.org}
}
\label{sec:biblio}
\appendix
\section{Technical Lemmas and Proofs Related to Differential Privacy of Morris Counter}
\label{append}
For the sake of completeness, we present here proofs of all technical lemmas that are not directly connected to Theorem \ref{thm:main}.
Some of computations are supported by Wolfram Mathematica ver.11.3 (\cite{Mathematica}). Whenever we obtain a result in this manner, we indicate it by $\WM$ sign. Usually results are precise, however in some cases, final forms are attained numerically.

We often struggle with expressions of a pattern $1-{1}/{y}$, so let us denote this function as $a(y)$ to abbreviate formulas.

Next two lemmas will be useful in a proof of Lemma \ref{lem:main}.
\begin{lemma}
\label{lem:anal1}
Let $c>{1}/{x}$. Then
$$
a(2cx)^{2y}\gs a(cx)^{y-1}\left(a(cx)+\frac{y}{4c^2x^2}\right)
$$
and
$$
a(cx)^y\gs a(2cx)^{2y-2}\left(a(2cx)^2-\frac{y}{4c^2x^2}\right)~.
$$
\end{lemma}

\begin{proof}

\[
\frac{a(2cx)^{2y}-a(cx)^y}{a(2cx)^2-a(cx)}=\sum_{i=0}^{y-1}a(2cx)^{2i} a(cx)^{y-i-1}~.
\]
Realize that the above denominator is $\left(1-\frac{1}{cx}+\frac{1}{4c^2x^2}-1+\frac{1}{cx}\right)=\frac{1}{4c^2x^2}$.
Hence, we obtained two inequalities: $a(2cx)^{2y}-a(cx)^y\gs \dfrac{y}{4c^2x^2} a(cx)^{y-1}$\\
and~$a(2cx)^{2y}-~a(cx)^y\ls~\dfrac{y}{4c^2x^2}~a(2cx)^{2(y-1)}$, which imply the thesis of this Lemma.
\end{proof}

\begin{lemma}
\label{lem:anal2}
Let $s\ls\log({x}/{4})$. Then $a(2^{-s}x)^{2x+1}<\exp(-2^{s+1})$ and 
$$
a(2^{-s}x)^{x-1}> \exp(-2^s) \left(1-\frac{2^{2 s-1}-2^s}{x}-\frac{2^{2s-7}+2^{4 s-3}}{x^2}\right)~.
$$
\end{lemma}

\begin{proof}
Let
$
f_1(x;s):=a(2^{-s}x)^{2x+1}$. For any $s\ls \log(x/4)$, we have $f_1(x;s)=\left(1-\frac{2^{2+1}}{2x}\right)^{2x=1}= \exp(-2^{s+1})\left(1 - O\left(x^{-1}\right)\right).
$
Realize a fact, that $z\ln(z)\gs z-1$, for $0<z\ls 1$. Hence

$$
\left(1-\frac{2^s}{x}\right)^{-2 x}\dfrac{\partial f_1(x;s)}{\partial x}= \frac{2^s (2 x+1)}{x^2}+2\left(1-\frac{2^s}{x}\right) \ln
   \left(1-\frac{2^s}{x}\right)>\frac{2^{s}}{x^2}>0
$$

and in a consequence $a(2^{-s}x)^{2x+1}<\exp(-2^{s+1})$ for any reasonable $s$.

Moreover, let us introduce
$$
D(x;s):=1-\frac{2^{2 s-1}-2^s}{x}-\frac{2^{2s-7}+2^{4 s-3}}{x^2}~.
$$
It is defined in such the way that $a(2^{-s}x)^{x-1}=D(x;s)+O(x^{-2})$.
Therefore we can attain:

$$
f_2(x;s):=\dfrac{a(2^{-s}x)^{x-1}}{D(x;s)}= \exp(-2^s)\left(1+O\left(x^{-2}\right)\right)~.
$$
Then, in a similar way
\begin{align*}
&D(x;s)^2 \left(1-\frac{2^s}{x}\right)^{-x+1} \dfrac{\partial f_2(x;s)}{\partial x}= \\
&D(x;s)\left(\frac{2^s (x-1)}{x^2\left(1-\frac{2^s}{x}\right)}+\ln\left(1-\frac{2^s}{x}\right)\right)-\left(\frac{2^{2 s-6}+2^{4s-2}}{x^3}+\frac{2^{2s-1}-2^s}{x^2}\right)\\
&<\left(\frac{2^{2s-1}-2^s}{x^2} + \frac{2^{3s}-2^{2s}}{x^3\left(1-\frac{2^s}{x}\right)}- \frac{2^{3s}}{3x^3}\right)	-\left(\frac{2^{2 s-6}+2^{4s-2}}{x^3}+\frac{2^{2s-1}-2^s}{x^2}\right)\\
&=\frac{2^{3s}-2^{2s}}{x^3\left(1-\frac{2^s}{x}\right)}- \frac{2^{3s}}{3x^3}-\frac{2^{2 s-6}+2^{4s-2}}{x^3}~.
\end{align*}
Let $d:=1-{2^s}/{x}$ and realize that $d\in[{3}/{4},1)$ and $2^s-1+d\left(-{2^s}/{3}-2^{-6}-2^{2s-2}\right)>0$. Indeed, if we put $z=2^s$, then we attain a quadratic inequality in $z$ variable, with determinant $\Delta=1-\frac{5d}{3}+\frac{55d^2}{576}$, that is negative for $d\in[{3}/{4},1)$.

Hence $\dfrac{\partial f_2(x;s)}{\partial x}<0$ and consequently
$$
a(2^{-s}x)^{x-1}> \exp(-2^s) \left(1-\frac{2^{2 s-1}-2^s}{x}-\frac{2^{2s-7}+2^{4 s-3}}{x^2}\right)
$$
for any reasonable $s$.
\end{proof}

\begin{lemma}
\label{lem:main}
\begin{enumerate}[a)]
\item The sequence $(p_{2^k +1,k+4})_{k=2}^{\infty}$ is descending.
\item The sequence $(p_{2^k +1,k+5})_{k=3}^{\infty}$ is ascending.
\end{enumerate}
\end{lemma}

\begin{proof}
Let $x=2^k$ and $t\in\{0,1\}$.
In advance we define 
$$
\kappa(k,t):=(-1)^{k+4+t} 2^{-\frac{(k+4+t)(k+3+t)}{2}} r_{k+4+t} 2^{-2x-1}
$$ 
and
$$
\tau(k,t):=[\![2\!\!\not|(k+t)]\!] (-1)^{k+t+3} 2^{-\frac{(k+t+3)(k+t+2)}{2}} 2r_{k+t+3} \left(\left(\frac{3}{4}\right)^{2x+1}\hspace{-4mm}-\left(\frac{1}{2}\right)^{x+2}\right)~,
$$
where $[\![\mathrm{cond}]\!]$ is the Iverson bracket of the condition $\mathrm{cond}$.

Realize that for $t\in\{0,1\}$ and $k\gs 5$, $|\tau(k,t)+\kappa(k,t)|<2^{-50}< 10^{-15}$.
Now, consider the differences between the consecutive elements of sequences:
\begin{align*}
&p_{2^{k+1} +1,k+5+t}-p_{2^k +1,k+4+t}\stackrel{\mathrm{Thm}~\ref{flajoletThm}}{=}\kappa(k,t)+\\
+&\sum_{i=0}^{k+3+t} (-1)^i 2^{-\frac{i(i-1)}{2}} r_i r_{k+t+4-i} \left[\left(1-\frac{2^{-5-t+i}}{x}\right)^{2x+1}-\left(1-\frac{2^{-4-t+i}}{x}\right)^{x+2}\right]\\
=&\sum_{i=0}^{\left\lfloor\frac{k+2+t}{2}\right\rfloor}\left\{ 2^{-i(2i-1)} r_{2i} r_{k+t+4-2i} \left[a(2^{5+t-2i}x)^{2x+1}-a(2^{4+t-2i}x)^{x+2}\right]\right.\\
-&\left.2^{-(2i+1)i} r_{2i+1} r_{k+t+3-2i} \left[a(2^{4+t-2i}x)^{2x+1} - a(2^{3+t-2i}x)^{x+2}\right]\right\}
+(\tau+\kappa)(k,t)\\
=&\hspace{-3mm}\sum_{i=0}^{\left\lfloor\frac{k+t+2}{2}\right\rfloor}\hspace{-2mm} 2^{-i(2i-1)} r_{2i+1} r_{k+t+4-2i} \left[ a(2^{2i+1}) \left(a(2^{5+t-2i}x)^{2x+1}-a(2^{4+t-2i}x)^{x+2}\right)\right. \\
-& \left. 2^{-2i} a(2^{4+t-2i}x) \left(a(2^{4+t-2i}x)^{2x+1} - a(2^{3+t-2i}x)^{x+2}\right)\right]
+\tau(k,t)+\kappa(k,t)~.
\end{align*}
Let us define $u_t:=2^{5+t-2i}$ (note that $u_t$ depends on $i$, but we abbreviate the notation for conciseness) and
\begin{align*}
W_t(i)&:=a(2^{2i+1}) \left(a(u_t x)^{2x+1}-a\left(\frac{u_t}{2} x\right)^{x+2}\right) \\
&- 2^{-2i} a\left(\frac{u_t}{2} x\right) \left(a\left(\frac{u_t}{2} x\right)^{2x+1} - a\left(\frac{u_t}{4} x\right)^{x+2}\right)
\end{align*}
and consider an upper bound of the last term:
\begin{align*}
W_t(i) &\stackrel{\mathrm{Lem.} \ref{lem:anal1}}{\ls} a(2^{2i+1}) \left(a(u_t x)^{2x+1}-a(u_t x)^{2x+2}\left(a(u_t x)^2-\frac{x+2}{u_t^2 x^2}\right)\right) \\
&- 2^{-2i} \left(a\left(\frac{u_t}{4} x\right)^{x}\left(a\left(\frac{u_t}{4} x\right)+\frac{x+1}{\frac{u_t^2}{4} x^2}\right) - a\left(\frac{u_t}{2} x\right) a\left(\frac{u_t}{4} x\right)^{x+2}\right)\\
&\WM a(2^{2i+1}) a(u_t x)^{2x+1}\frac{1}{x}\left(\frac{3u_t +1}{u_t^2}-\frac{u_t+1}{u_t^3 x}-\frac{1}{u_t^3 x^2}\right) \\
&- 2^{-2i} a\left(\frac{u_t}{4} x\right)^{x}\frac{1}{x}\left(\frac{6u_t+4}{u_t^2}-\frac{28}{x u_t^2}+\frac{32}{x^2 u_t^3}\right)~.
\end{align*}
\mbox{ }\bigskip
Note that $2i\ls k+2+t$, so ${8}/{x}\ls u_t$ and in consequence:
\begin{itemize}
    \item[] \quad $6u_t-{28}/{x}>{20}/{x}>0$,
    \item[]
\begin{align*}
 u_t\cdot(3u_t+1)-\frac{u_t+1}{x}-\frac{1}{x^2}&\gs u_t\cdot\left(\frac{24}{x}+1\right)-\frac{u_t+1}{x}-\frac{1}{x^2}\\
&\gs \frac{8}{x} - \frac{1}{x} + \frac{23u_t}{x} -\frac{1}{x^2}\gs \frac{7}{x} +\frac{183}{x^2}>0~.
\end{align*}
\end{itemize}

Hence
\begin{align}
&W_t(i)\stackrel{\mathrm{Lem.} \ref{lem:anal2}}{<} a(2^{2i+1}) \exp\left(-\frac{2}{u_t}\right)\frac{1}{x}\left(\frac{3u_t+1}{u_t^2}-\frac{u_t+1}{x u_t^3}-\frac{1}{x^2 u_t^3}\right) \nonumber\\
&- 2^{-2i} a\left(\frac{u_t}{4} x\right) \exp\left(-\frac{4}{u_t}\right)D(x;2i-3-t)\frac{1}{x}\left(\frac{6u_t+4}{u_t^2}-\frac{28}{x u_t^2}+\frac{32}{x^2 u_t^3}\right) \nonumber\\
&\WM a(2^{2i+1}) \exp\left(-\frac{2}{u_t}\right)\frac{1}{x}\left(\frac{3u_t+1}{u_t^2}-\frac{u_t+1}{x u_t^3}-\frac{1}{x^2 u_t^3}\right) \label{eq:utx}\\
&-\frac{2^{-2i}}{x}\! \exp\left(-\frac{4}{u_t}\right)\!\!
\left(\frac{6u_t+4}{u_t^2} -\frac{32+48u_t+28u_t^2}{u_t^4 x}-\frac{128 + 64u_t - \frac{703}{2}u_t^2 + \frac{259}{4}u_t^3}{u_t^6 x^2}\right. \nonumber\\
&+\left. \frac{512 + 128 u_t - 1150 u_t^2 + \frac{909}{2} u_t^3}{u_t^7 x^3} - \frac{4608 - 1024 u_t + 530 u_t^2}{u_t^7 x^4}+\frac{4096+16u_t^2}{u_t^8 x^5}\right) \nonumber
\end{align}
We denote the upper bound (\ref{eq:utx}) by $U_t(x;u_t(i))$.

Analogically we would like to establish a lower bound of $W_t(i)$:
\begin{align}
W_t(i)\stackrel{\mathrm{Lem.} \ref{lem:anal1}}{\gs}& a(2^{2i+1}) \left(a(u_t x)a\left(\frac{u_t}{2} x\right)^{x-1}\left(a\left(\frac{u_t}{2} x\right)+\frac{1}{u_t^2 x}\right)-a\left(\frac{u_t}{2} x\right)^{x+2}\right) \nonumber\\
-& 2^{-2i} a\left(\frac{u_t}{2} x\right)\!\left(a\left(\frac{u_t}{2} x\right)^{2x+1}\! - a\left(\frac{u_t}{2} x\right)^{2x+2}\left(a\left(\frac{u_t}{2} x\right)^2 -\frac{x+2}{\frac{u_t^2}{4} x^2}\right)\right)\nonumber\\
\WM& a(2^{2i+1}) a\left(\frac{u_t}{2} x\right)^{x-1}\frac{1}{x}\left(\frac{3u_t+1}{u_t^2}- \frac{10u_t+1}{u_t^3 x}+ \frac{8}{u_t^3 x^2}\right) \label{eq:techno}\\
-& 2^{-2i} a\left(\frac{u_t}{2} x\right)^{2x+2}\frac{1}{x}\left(\frac{6u_t+4}{u_t^2} -\frac{4u_t+8}{u_t^3 x}- \frac{8}{u_t^3 x^2}\right) \nonumber
\end{align}
Now from ${8}/{x}\ls u_t$ we attain 
\begin{align}
u_t\cdot(3u_t+1)-\frac{10u_t+1}{x}&=u_t\cdot\left(\frac{7u_t}{4}+\frac{7}{8}\right)+(10u_t+1)\left(\frac{u_t}{8}-\frac{1}{x}\right)\nonumber\\
&\gs u_t\cdot\left(\frac{7u_t}{4}+\frac{7}{8}\right)>0 \label{eq:uteki}
\end{align}
and 
$$
u_t\cdot(6u_t+4)-\frac{4u_t+8}{x}-\frac{8}{x^2}\gs \frac{48u_t}{x}+\frac{32}{x}-\frac{4u_t+8}{x}-\frac{8}{x^2}\gs \frac{24}{x} + \frac{344}{x^2}>0~.
$$
Hence
\begin{align}
&W_t(i) \stackrel{\mathrm{Lem.} \ref{lem:anal2}}{>} \! a(2^{2i+1}) \exp\left(-\frac{2}{u_t}\right)\! \frac{D(x;2i-4-t)}{x}\left(\frac{3u_t+1}{u_t^2}- \frac{10u_t+1}{u_t^3 x}+ \frac{8}{u_t^3 x^2}\right) \nonumber \\
&- 2^{-2i} \exp\left(-\frac{4}{u_t}\right) a\left(\frac{u_t}{2} x\right)\frac{1}{x}\left(\frac{6u_t+4}{u_t^2} -\frac{4u_t+8}{u_t^3 x}- \frac{8}{u_t^3 x^2}\right) \nonumber\\
&\WM \frac{a(2^{2i+1})}{x} \exp\left(-\frac{2}{u_t}\right) \left(\frac{3u_t+1}{u_t^2}-\frac{2+5u_t+4u_t^2}{u_t^4 x}-\frac{2+4u_t-\frac{575}{32} u_t^2+ \frac{387}{32}u_t^3}{u_t^6 x^2}\right. \nonumber\\
&\left.+ \frac{2+20 u_t -\frac{511}{32} u_t^2 +\frac{261}{16} u_t^3}{u_t^7 x^3}- \frac{16 + \frac{1}{4}u_t^2}{u_t^7 x^4}\right) \label{eq:ltx}\\
&- 2^{-2i} \exp\left(-\frac{4}{u_t}\right)\frac{1}{x}\left(\frac{6u_t+4}{u_t^2}-\frac{16u_t+16}{u_t^3 x}+\frac{16}{u_t^4 x^2}+\frac{16}{u_t^4 x^3}\right)~. \nonumber
\end{align}
Denote the lower bound (\ref{eq:ltx}) by $L_t(x;u_t(i))$.

Now we show that $W_t(i)>0$ for $i\gs 1$.
Indeed, after reducing the redundant terms from inequality (\ref{eq:techno}), together with inequality (\ref{eq:uteki}), we can obtain
\begin{equation}
\label{eq:lower}
W_t(i)>\frac{a\left(\frac{u_t}{2} x\right)^{x}}{xu_t^2} \left(a(2^{2i+1})\frac{14u_t+7}{8}- 2^{-2i} (6u_t+4)\right)
\end{equation}
Let us denote $C(t,x):=\frac{a\left(\frac{u_t}{2} x\right)^{x}}{xu_t^2}$. In case $i\gs 2$, we may attain $$(\ref{eq:lower})\gs \frac{C(t,x)}{256}\left(31(14u_t+7) - 96u_t-64\right)>0.$$

When $i=1$, then $u_t\gs 8$, so
$(\ref{eq:lower})\gs \frac{C(t,x)}{64}\left(7(14u_t+7)- 96u_t-64\right)=\frac{2u_t-15}{64}\gs\frac{1}{64}$.

Thanks to the property $W_t(i)>0$ for $i\gs 1$, we may subtly neutralize the influence of $r_{k+5-i}$ in the considered sum:
$$
\sum_{i=0}^{\left\lfloor\frac{k+2}{2}\right\rfloor} 2^{-i(2i-1)} r_{2i+1}r_{k+5-i} W_0(i)<r_{k+5}\sum_{i=0}^{\left\lfloor\frac{k+2}{2}\right\rfloor} 2^{-i(2i-1)} r_{2i+1} W_0(i)~.
$$
Naturally we may consider $U_0(x;u_0(i))$ instead of $W_0(i)$ numerically for $i\ls 4$:
\begin{multline*}
\sum_{i=0}^{4} 2^{-i(2i-1)} r_{2i+1} U_0(x;u_0(i))\WM -8.294491525704523\ldots\cdot 10^{-6}+\frac{0.15588\ldots}{x}\\
+\frac{0.00407163\ldots}{x^2}-\frac{0.0298032\ldots}{x^3}+\frac{0.0198815\ldots}{x^4}-\frac{0.00785419\ldots}{x^5}~,
\end{multline*}
so for $x\gs 2^{15}\;(k\gs 15)$, $\sum\limits_{i=0}^{4} 2^{-i(2i-1)} r_{2i+1} W_0(i)\ls -3.53741\cdot 10^{-6}$.
Moreover we may bound $W_0(i)$ by $a(2^{5-2i}x)^{2x+1}$ for the rest of the sum:
$$
\sum_{i=5}^{\left\lfloor\frac{k+2}{2}\right\rfloor} 2^{-i(2i-1)} r_{2i+1} a(2^{5-2i}x)^{2x+1}\ls  \frac{R\; 2^{-45}\exp(-64)}{1-2^{-21} \exp(-192)}= 1.5784\ldots\cdot 10^{-41}~,
$$
so $p_{2^{k+1} +1,k+5}-p_{2^k +1,k+4}<0$ for $k\gs 15$.\\
However, according to Theorem \ref{flajoletThm}, we also present the numerical values of the sequence $(p_{2^k +1,k+4})_{k=2}^{14}$ in the Table \ref{tab:init}. We can now easily see that for any $k\gs 2$ we attained $p_{2^{k+1} +1,k+5}-p_{2^k +1,k+4}<0$ .
\begin{table}[h!]
\centering
\begin{tabular}{|c|c||c|c||c|c|}
\hline
k & $p_{2^k +1,k+4}$ & k & $p_{2^k +1,k+4}$ & k & $p_{2^k +1,k+4}$ \\
\hline\hline
2 & $0.0000305176\ldots$ & 7 & $0.0000189841\ldots$ & 12 & $0.0000185484\ldots$ \\
\hline
3 & $0.0000256707\ldots$ & 8 & $0.0000187590\ldots$ & 13 & $0.0000185413\ldots$ \\
\hline
4 & $0.0000221583\ldots$ & 9 & $0.0000186466\ldots$ & 14 & $0.0000185378\ldots$ \\
\hline
5 & $0.0000203424\ldots$ & 10 & $0.0000185904\ldots$ & & \\ 
\hline
6 & $0.0000194356\ldots$ & 11 & $0.0000185624\ldots$ & & \\
\hline
\end{tabular}
\caption{Numerical values of the sequence $(p_{2^k +1,k+4})_{k=2}^{14}$.}
\label{tab:init}
\end{table}

Moreover, realize that ${r_{k+5}}/{r_{k+3}}<1.1$ for any $k\gs 3$, so 
\begin{align*}
\sum\limits_{i=0}^{1} 2^{-i(2i-1)} r_{2i+1} 1.1^{1-i} L_1(x;u_1(i)) \WM 0.00128843\ldots +\frac{0.00212699\ldots}{x}\\
-\frac{0.00326251\ldots}{x^2}+\frac{0.000219133}{x^3}-\frac{3.50924875\ldots\cdot10^{-7}}{x^4}
\end{align*}
For any possible $x\gs 8$ ($k\gs 3$), $\sum_{i=0}^{1} 2^{-i(2i-1)} r_{2i+1} 1.1^{1-i} L_1(x;u_t(i))> 0.0015$.
We already know that $W_1(i)$ are positive for $i>1$, so
$p_{2^{k+1} +1,k+6}-p_{2^k +1,k+5}>0$ for all $k\gs 3$.
\end{proof}

We may use Theorem \ref{flajoletThm} once again to see that $$\frac{p_{2^6+1,10}}{p_{2^6+1,11}}= 129.454\ldots>2^{7} \mbox{and},$$ $$\frac{p_{2^7+1,11}}{p_{2^7+1,12}}= 125.065\ldots<2^{7}.$$ Together with Lemma \ref{lem:main} we may easily attain Claim \ref{claim4} and we instantly see that this Claim cannot be extended continuously for $k<7$.

\begin{lemma}
\label{lem:between}
Let $2\ls l\ls n$ and assume that 
$\alpha_i=2^{i-2}\frac{p_{n,l-i}}{p_{n,l-i+1}}$ for $i\in[0:2]$ and
${\alpha'}_j=2^{j-2} \frac{p_{n+1,l-j}}{p_{n+1,l-j+1}}$ for $j\in [0:1]$.
\\
If $0\ls\alpha_2<\alpha_1<\alpha_0$, then $0<\alpha'_1<\alpha'_0$.
\end{lemma}
\begin{proof}
Realize that $p_{n+1,l-i+1}=p_{n,l-i+1}(1-2^{-l+i-1}+2^{-l+2}\alpha_i)$ for $i\in[0:2]$, so for $j\in[0:1]$,
\begin{align*}
\alpha'_j&=\frac{p_{n+1,l-j}}{2^{2-j}p_{n+1,l-j+1}}=\frac{p_{n,l-j}(1-2^{-l+j}+2^{-l+2}\alpha_{j+1})}{2^{2-j}p_{n,l-j+1}(1-2^{-l+j-1}+2^{-l+2}\alpha_j)}\\
&=\frac{\alpha_j(1-2^{-l+j}+2^{-l+2}\alpha_{j+1})}{1-2^{-l+j-1}+2^{-l+2}\alpha_j}~.
\end{align*}
Assume that $\alpha'_1\gs \alpha'_0$. Then 
$$
A:=\alpha_1(1-2^{-l+1}+2^{-l+2}\alpha_2)(1-2^{-l-1}+2^{-l+2}\alpha_0)\gs\alpha_0(1-2^{-l}+2^{-l+2}\alpha_1)^2=:B.
$$
However, contrary to the assumption,
\begin{align*}
A&=\alpha_1(1-2^{-l+1}+2^{-2l}-2^{-l-1}+2^{-l+2}(\alpha_0+\alpha_2)-2^{-2l+3}\alpha_0\\
&-2^{-2l+1}\alpha_2+2^{-2l+4}\alpha_0 \alpha_2)\\
&<\alpha_0(1-2^{-l+1}+2^{-2l})+\alpha_1(2^{-l+2}(2\alpha_0)+2^{-2l+4}\alpha_0 \alpha_1)\\
&<\alpha_0(1-2^{-l+1}+2^{-2l}+\alpha_1(2^{-l+3}+2^{-2l+3}+2^{-2l+4}\alpha_1))=B~.
\end{align*}

\end{proof}

\begin{lemma}
\label{lem:end}
If for some $n\in\NN$, $\eta_{n}=2^{-n}\frac{p_{n,n}}{p_{n,n+1}}$ and\\ $\eta_{n+1}=2^{-n-1}\frac{p_{n+1,n+1}}{p_{n+1,n+2}}$, then $\eta_n<\eta_{n+1}$.
\end{lemma}

\begin{proof}
\begin{align*}
0&=p_{n+1,n+1}-2^{n+1} \eta_{n+1} p_{n+1,n+2}=p_{n,n+1}(1-2^{-n-1})+p_{n,n}2^{-n}\\
&-\eta_{n+1} p_{n,n+1}=p_{n,n+1}(1-2^{-n-1}+\eta_n-\eta_{n+1})~,
\end{align*}
but $1-2^{-n-1}>0$, so $\eta_n<\eta_{n+1}$.
\end{proof}

\end{document}